\documentclass[article]{IEEEtran}
\usepackage{epsfig,latexsym,amsmath,amssymb,setspace,cite,color,graphicx,algorithm,algorithmic,cases}
\usepackage[caption=false,font=footnotesize]{subfig}
\usepackage[percent]{overpic}

\usepackage{amsmath}
\usepackage{amsfonts}
\usepackage{amssymb}
\usepackage{dsfont}
\usepackage{bm}
\usepackage{amsthm}
\usepackage{newlfont}
\usepackage{float}
\usepackage{hyperref}
\usepackage{algorithm}
\usepackage{algorithmic}
\usepackage{enumerate}
\usepackage{chngcntr}
\usepackage{mathtools}
\usepackage{breqn}
\usepackage{stmaryrd}
\usepackage{cite}
\usepackage[yyyymmdd]{datetime}
\hypersetup{
    colorlinks=true, 
    linkcolor= blue, 
    citecolor=blue 
}

\begin{document}
\sloppy
\allowdisplaybreaks[1]

\newtheorem{thm}{Theorem} 
\newtheorem{lem}{Lemma}
\newtheorem{prop}{Proposition}
\newtheorem{cor}{Corollary}
\newtheorem{defn}{Definition}
\newcommand{\remarkend}{\IEEEQEDopen}
\newtheorem{remark}{Remark}
\newtheorem{rem}{Remark}
\newtheorem{ex}{Example}
\newtheorem{pro}{Property}

\newenvironment{example}[1][Example]{\begin{trivlist}
\item[\hskip \labelsep {\bfseries #1}]}{\end{trivlist}}

\renewcommand{\qedsymbol}{ \begin{tiny}$\blacksquare$ \end{tiny} }

\renewcommand{\leq}{\leqslant}
\renewcommand{\geq}{\geqslant}

\title {The Gaussian Multiple Access Wiretap Channel with Selfish Transmitters:\\ A Coalitional Game Theory Perspective} 

\author{\IEEEauthorblockN{R\'emi A. Chou and  Aylin Yener}

\thanks{R\'{e}mi A. Chou is with the Department of Computer Science and Engineering, The University of Texas at Arlington, Arlington, TX 76019. Aylin Yener is with the Department of Electrical and Computer Engineering, The Ohio State University, Columbus, OH 43210. E-mails: remi.chou@uta.edu, yener@ece.osu.edu. A preliminary version of this work was presented to the 2017 IEEE International Symposium on Information Theory (ISIT) in \cite{ISIT173}. This work was supported in part by NSF grants CIF-1319338, CNS-1314719, and CCF-2401373. 
}
}

\maketitle

\begin{abstract}
This paper considers the Gaussian multiple access wiretap channel (GMAC-WT) with selfish transmitters, i.e., who are each solely interested in maximizing their individual secrecy rate. The question then arises as to whether selfish transmitters can increase their individual secrecy rate by participating in a collective, i.e., multiple access, protocol instead of operating on their own. If yes, the question arises whether there is a protocol that satisfies all the participating transmitters simultaneously, in the sense that no transmitter has an incentive to deviate from the protocol. 
Utilizing coalitional game theory, these questions are addressed for the degraded GMAC-WT with an arbitrary number of transmitters and for the non-degraded GMAC-WT with two transmitters. In particular, for the degraded GMAC-WT, cooperation is shown to be in the best interest of all transmitters, and the existence of protocols that incentivize all transmitters to participate is established. Furthermore, a unique, fair, stable, and achievable secrecy rate~allocation is determined. For the non-degraded GMAC-WT, depending on the channel parameters, there are cases where cooperation is not in the best interest of all transmitters, and cases where it is. In the latter cases, a unique, fair, stable, and achievable secrecy rate~allocation is determined.
\end{abstract} 

\section{Introduction}

We study secure communication over a Gaussian multiple access wiretap channel (GMAC-WT)~\cite{tekin2008gaussian,tekin2008general}. 
The information-theoretic problem formulation for secure communication over the GMAC-WT enables establishing fundamental limits of achievable rate-tuples, under the assumption of altruistic legitimate entities. That is, the underlying assumption is one of full cooperation where transmitters work together to achieve the largest secure rate region. An equally valid scenario could be that the transmitters are interested only in maximizing {\it their individual} secure rates. This can lead to a conflict of interests and fairness issues among transmitters as they try to capture limited resources for their own benefit.  Only certain rate-tuples, if any, would be acceptable by selfish transmitters. 
A large body of the literature has considered similar questions in multiuser communication problems by means of game theory, see, for instance,  \cite{la2004game,gajic2008game,zhu2009constrained,perlaza2013equilibria,ge2015secure,amor2016decentralized} for  Gaussian multiple access channels and~\cite{leshem2008cooperative,mathur2006coalitional,berry2011shannon,liu2011game,jorswieck2008complete,larsson2008competition,rose2011nash,quintero2017nash,quintero2018approximate} for interference channels. We also refer to~\cite{saad2009coalitional,lasaulce2011game,han2012game}  and references therein for the treatment of broader classes of multiuser communication problems.

Our contribution can be summarized as follows. (i) We cast the problem of selfish transmitters over the GMAC-WT as a coalitional game~\cite{peleg2007introduction,osborne1994course,myerson2013game}  in which the  value function is determined under information-theoretic
guarantees, i.e., the value associated with a coalition is computed with no restrictions on the strategies that the transmitters outside the coalition can adopt. (ii) For the degraded GMAC-WT, we show that there exist collective protocols for which no transmitter has an incentive to deviate, i.e., it is in the best interest for selfish transmitters to collaborate. In particular, we show that the core of the game we have defined is non-empty and intersects known achievable regions for the GMAC-WT. Using an axiomatic solution concept, we determine a unique, fair, stable, and achievable secrecy rate allocation. (iii)~For the non-degraded GMAC-WT with two transmitters, we show that, depending on the channel parameters, cooperation may or may not in the best interest of all transmitters. When cooperation is in the best interest of all transmitters, we generalize the axiomatic solution concept used in the degraded case and determine a unique, fair, stable, and achievable secrecy rate allocation.

\subsection{Related work}
\emph{Related work on the GMAC-WT}:
The GMAC-WT has first been introduced in \cite{tekin2008gaussian} for degraded channels to account for the presence of an eavesdropper at the physical layer in multiple access communication, and to provide information-theoretic security guarantees against such an adversary.  This model has further been studied in the case of non-degraded channels in~\cite{tekin2008general}. \cite{tekin2008gaussian,tekin2008general} and most subsequent works focused on characterizing the capacity region for this model under the assumption that all the transmitters are willing to participate in a joint protocol. By contrast, in this paper, we assume that the transmitters are selfish, i.e., will only agree to participate in a joint multiuser protocol if it benefit them through a higher communication rate. We treat this problem by means of cooperative game theory. Note that the {GMAC-WT} with selfish transmitters is  also considered  in \cite{ge2015secure} but via non-cooperative game theory and with the assumption that all transmitters follow a pre-determined transmission strategy. Our paper is also related to \cite{banawan2016achievable}, as we characterize the worst behavior that a group of transmitters can adopt to prevent confidential communication between the other transmitters and the legitimate receiver. Our study contrasts with~\cite{ge2015secure,banawan2016achievable}, as we allow \emph{any communication strategy} for the transmitters that might be unwilling to participate in a collective protocol.

\emph{Other related work}: In this paper, we treat the problem of selfish transmitters over the \mbox{GMAC-WT} by means of a coalitional game theory framework. We refer to~\cite{peleg2007introduction,osborne1994course,myerson2013game} for an introduction to coalitional game theory, and to \cite{saad2009coalitional} for a review of some of its applications to telecommunications. 
The coalitional game we define is inspired by the game formulation of~\cite{la2004game} for the Gaussian multiple access channel in the absence of an eavesdropper and thus security constraints. Hence, the setting in \cite{la2004game} is recovered as a special case of our game. Our results show that, similar to the Gaussian multiple access channel, for the degraded Gaussian multiple access wiretap channel, the participation of all transmitters in a joint protocol is in their best interest to maximize their secrecy rate. However, our study contrasts this result by showing that, for non-degraded Gaussian multiple access wiretap channels, this is no longer always the case. Specifically, for a two-transmitter non-degraded multiple access wiretap channel, we determine sufficient conditions that ensure 
 that the participation of all the transmitters in a joint protocol is in their best interest. Additionally, we provide an example of a non-degraded channel for which the participation of both transmitters in a joint protocol is not in their best interest.  Finally, note that a coalitional game theory approach to information-theoretic security has also been used in the context of many-to-one secret key generation in~\cite{ISIT172}. However, as a main difference, the determination of the value function of the game in \cite{ISIT172} relies on information disclosure threats, whereas it relies on jamming threats in this paper. Additionally, unlike in \cite{ISIT172}, the resulting game in the present study is not convex~\cite{shapley1971cores}, which significantly complexifies its study, and a natural axiomatic definition of fairness can be formulated and studied.

\subsection{Organization of the paper}
The remainder of the paper is organized as follows. Before we present our model for the GMAC-WT with selfish transmitters, we review a pivotal auxiliary problem: the GMAC-WT with adversarial jammers in Section~\ref{sec:modelmac}. This auxiliary problem allows us to understand what is the worst case case scenario when transmitters refuse to cooperate. Our main problem, the GMAC-WT with selfish transmitters is treated in Section~\ref{sec:game}, which deals with the degraded case, and in Section~\ref{sec:game2}, which deals with the non-degraded case. We end the paper with concluding remarks in Section~\ref{sec:concl}.

\subsection{Notation}\label{sec:notation}
Throughout the paper, define $\llbracket a, b\rrbracket \triangleq [\lfloor a \rfloor , \lceil b \rceil] \cap \mathbb{N}$. The components of a vector, $X^{n}$, of size $n\in\mathbb{N}$, are denoted by subscripts, i.e., $X^{n} \triangleq (X_1 , X_2, \ldots, X_{n})$. 
 For $x \in \mathbb{R}$,  define $[x]^+ \triangleq \max(0,x)$.
The power set of $\mathcal{S}$ is denoted by $2^{\mathcal{S}}$. Unless specified otherwise, capital letters designate random variables, whereas lowercase letters designate realizations of associated random variables, e.g., $x$ is a realization of the random variable $X$. For $R \geq 0$, $n \in\mathbb{N}^*$, $\mathbb{B}^n_0(R)$ denotes the ball of radius $R $ centered in $0$ in $\mathbb{R}^n$ under the Euclidian norm. For any set $\mathcal{S} \subset \mathbb{N}$, and any sequence $(R_s)_{s \in \mathcal{S}}$ of real numbers, the notation $R_{\mathcal{S}}$ denotes the sum $\sum_{s \in \mathcal{S}}R_s$. 
\section{Auxiliary problem: Gaussian multiple access wiretap channel with adversarial jammers} \label{sec:modelmac}
We review in this section the GMAC-WT with adversarial jammers.  We present the model and review known results in Sections \ref{sec:model} and \ref{sec:res}, respectively. This auxiliary problem will be used in our model for the GMAC-WT with selfish transmitters in Sections \ref{sec:game}, \ref{sec:game2}.

\subsection{Model} \label{sec:model}
Reference \cite{tekin2008gaussian} introduces and studies a degraded Gaussian multiple access wiretap channel with several transmitters in the presence of an eavesdropper. We consider a similar Gaussian multiple access wiretap channel with additional adversarial jammers that help the eavesdropper to minimize the secrecy rates between the legitimate transmitters and the receiver. Specifically, we assume that the collective signal emitted by the jammers is known by the eavesdropper, who is able to cancel it out from its observations. We also assume that the collective signal emitted by the jammers is prescribed by a power constraint. However, the jamming strategy is arbitrary   and unknown to the legitimate transmitters and the receiver.  
 Such modeling is introduced in \cite{chou21} and follows the works in~\cite{csiszar1991capacity} for the Gaussian arbitrarily varying  channel and in \cite{la2004game} for the Gaussian arbitrarily varying multiple access channel. 
 
For completeness, we review the model and specific results of interests from \cite{chou21} for our study.
 In the remainder of the paper, we let $\mathcal{L} \triangleq \llbracket 1, L \rrbracket$ denote the set of transmitters. We consider the following channel model,  
 \begin{subequations}
\begin{align}
Y^n & \triangleq \sum_{l \in \mathcal{L}} X_l^n + S^n + N_Y^n, \label{eqmod1}\\
Z^n& \triangleq \sum_{l \in \mathcal{L}} \sqrt{h_l} X_l^n + N_{Z}^n, \label{eqmod2}
\end{align}
\end{subequations}
where $Y^n$ is the channel output observed by the legitimate receiver, $Z^n$ is the modified channel output observed by the eavesdropper after cancellation of the jamming signal $S^n$, $S^n$ is an arbitrary jamming sequence emitted by the eavesdropper satisfying the power constraint $ \Vert S^n \Vert^2 \triangleq \sum_{i=1}^n S_i^2 \leq n \Lambda$ , for $l\in\mathcal{L}$, $X^n_l$ is the signal emitted by transmitter $l$ satisfying the power constraint $ \Vert X_l^n \Vert^2 \triangleq \sum_{i=1}^n X_i^2 \leq n \Gamma_l$, and $N_{Y}^n$ and $N_{Z}^n$ are sequences of  independent and identically distributed Gaussian noises with unit variances $\sigma^2_Y$, $\sigma^2_Z$, respectively.   
We refer to this model as the Gaussian multiple access wiretap  channel with adversarial jammers (GMAC-WT-AJ) with parameters $((\Gamma_l)_{l\in \mathcal{L}},(h_l)_{l \in \mathcal{L}}, \Lambda, \sigma^2_Y, \sigma^2_Z)$. When the channel gains $(h_l)_{l \in \mathcal{L}}$ are all equal to $h \in [0,1[$, we refer to this model as the degraded GMAC-WT-AJ with parameters $((\Gamma_l)_{l\in \mathcal{L}},h, \Lambda, \sigma^2_Y, \sigma^2_Z)$.

We define a coding scheme and achievable rates for our channel model following the scheme over multiple encoding blocks of~\cite{la2004game} to allow time-sharing. 
\begin{defn} \label{def:WT1}
Let $n,k \in \mathbb{N}$. A $ \left( (2^{nR_l})_{l\in \mathcal{L}}, n,k\right)$ code $\mathfrak{C}_n$ for the GMAC-WT-AJ
 consists for each $i \in \llbracket 1 , k \rrbracket$ of
\begin{itemize}
\item $L$ messages sets $\mathcal{M}^{(i)}_l \triangleq \llbracket 1 , 2^{nR^{(i)}_l} \rrbracket$, $l \in \mathcal{L}$;
\item $L$ stochastic encoders, $f^{(i)}_{l} : \mathcal{M}_l \to \mathbb{B}_0^n(\sqrt{n\Gamma_l}) $, $l \in \mathcal{L}$, which maps a uniformly distributed message $M^{(i)}_l \in  \mathcal{M}^{(i)}_l$ to a codeword of length $n$;
\item One  decoder,  $g^{(i)} : \mathbb{R}^n \to \bigtimes_{l \in \mathcal{L}} \mathcal{M}^{(i)}_l$, which maps a sequence of $n$ channel output observations
to an estimate $\left( \widehat{M}^{(i)}_l \right)_{l \in \mathcal{L}}$ of the messages $ \left(M^{(i)}_l\right)_{l \in \mathcal{L}}$;
\end{itemize}
where for any $l \in \mathcal {L}$, $R_l \triangleq \frac{1}{k} \sum_{i=1}^k R^{(i)}_l$, and operates as follows.
For each $i \in \llbracket 1, k \rrbracket$, Transmitter $l \in \mathcal{L}$, encodes the message  $M^{(i)}_l$ with $f_l^{(i)}$, and sends the encoded message to the legitimate receiver over the channel defined by \eqref{eqmod1}, \eqref{eqmod2} with power constraint $n \Lambda $ for the jamming signal $S^n_i$.  The legitimate receiver forms from his observations the estimate of $\left( \widehat{M}^{(i)}_l \right)_{l \in \mathcal{L}}$ of the messages $ \left(M^{(i)}_l\right)_{l \in \mathcal{L}}$. We define $\widehat{M}_{\mathcal{L}} \triangleq \left( \widehat{M}^{(i)}_l \right)_{l \in \mathcal{L}, i \in \llbracket 1 , k \rrbracket}$, ${M}_{\mathcal{L}} \triangleq \left(M^{(i)}_l\right)_{l \in \mathcal{L}, i \in \llbracket 1 , k \rrbracket}$, $S \triangleq (S^n_i)_{i \in \llbracket 1 , k \rrbracket}$, $\mathcal{S} \triangleq \{ (S^n_i)_{i \in \llbracket 1 , k \rrbracket} : \lVert S^n_i \lVert^2 \leq n \Lambda, \forall i \in \llbracket 1 , k \rrbracket \}$. 
\end{defn}

\begin{defn} \label{def:WT2} 
A rate tuple $(2^{nR_l})_{l\in \mathcal{L}}$ is achievable, if there exists a sequence of $( (2^{nR_l})_{l\in \mathcal{L}}, n,k)$ codes $\mathfrak{C}_n$ for the GMAC-WT-AJ 
such~that
\begin{subequations}
\begin{align}
\lim_{n \to \infty} \sup_{S \in \mathcal{S}}  \mathbb{P}[ \widehat{M}_{\mathcal{L}} \neq {M}_{\mathcal{L}}] =0& \text{ (reliability)},\\
\lim_{n \to \infty} \frac{1}{nk} H( {M}_{\mathcal{L}}|{Z}^{kn} ) \geq \sum_{l \in\mathcal{L}} R_l & \text{ (equivocation)}. \label{eq:sec}
\end{align}
\end{subequations}
\end{defn}
 
We assume the transmitters selfish. Hence, a  transmitter that cannot transmit at a positive secrecy rate will preserve power, i.e., cooperative jamming \cite{tekin2008general} is ruled out. 
Note that the model described in Definitions \ref{def:WT1} and \ref{def:WT2} recovers the model introduced in \cite{la2004game} in the absence of the security constraint~\eqref{eq:sec}.

\subsection{Review of known results} \label{sec:res}
Given $\Lambda \in \mathbb{R}_+$ and $(\Gamma_l)_{l\in \mathcal{L}}$, we define $h_{\Lambda} = (1+\Lambda)^{-1}$, $\mathcal{L}(\Lambda) \triangleq \{ l \in \mathcal{L}: \Gamma_l > \Lambda\}$, and $
\mathcal{L}^c(\Lambda) \triangleq \mathcal{L} \backslash {\mathcal{L}(\Lambda)}.$

The following theorems provides an achievability region and a a sum-rate capacity result for the problem defined in Section~\ref{sec:model}.

\begin{thm}[\cite{chou21}]\label{thregion}
The following region is achievable for the degraded GMAC-WT-AJ with parameters $((\Gamma_l)_{l\in \mathcal{L}},h, \Lambda, 1, 1)$  
 \begin{align}
&\smash{\mathcal{R} = \!\!\!\!\bigcup_{\substack{(P_l)_{l\in \mathcal{L}} \\: \forall l\in \mathcal{L}(\Lambda), \Lambda < P_l \leq \Gamma_l
}} \!\!\!\! \left\{  \!\!\!\!\!\!\!\!\!\!\! \phantom{\frac{\sum_{\mathcal{T}}}{\Lambda}} (R_l)_{l \in \mathcal{L}} :  \forall l \in\mathcal{L}^c(\Lambda),  R_l = 0 \text{ and }  \right. } \nonumber \\
& \phantom{----------}    \nonumber\\ & \left. \phantom{}\forall \mathcal{T} \subseteq \mathcal{L}(\Lambda),R_{\mathcal{T}} \leq \left[ \frac{1}{2} \log \left( \frac{1+ h_{\Lambda}P_{\mathcal{T}}}{1+ h P_{\mathcal{T}}(1+hP_{\mathcal{T}^c})^{-1}} \right)     \right]^+  \right\}. \label{eq:reg}
\end{align}
\end{thm}
We also have the following optimality result. 
\begin{thm}[\cite{chou21}]\label{thsumrate}
The maximal secrecy sum-rate $R_{\mathcal{L}}\triangleq \sum_{l\in \mathcal{L}}R_l$ achievable for the degraded GMAC-WT-AJ with parameters $((\Gamma_l)_{l\in \mathcal{L}},h, \Lambda, 1, 1)$ is
\begin{align} \label{eqsumrate}
 \left[ \frac{1}{2} \log \left( \frac{ 1+ h_{\Lambda} \Gamma_{\mathcal{L}(\Lambda)} }{1+  h\Gamma_{\mathcal{L}(\Lambda)}} \right)    \right]^+.
\end{align}
\end{thm}
Note that the optimal secrecy sum-rate is positive if and only if $h_{\Lambda} > h$ and $\mathcal{L}(\Lambda) \neq \emptyset$.

\section{Degraded GMAC-WT with selfish transmitters} \label{sec:game}

We define in Section \ref{secmod2} a coalitional game for the GMAC-WT when the transmitters are assumed selfish. We study the properties of the game and its core in Section \ref{sec:supadd}. In Section~\ref{sec:alloc}, we propose a solution concept for a fair allocation, and determine a unique solution that corresponds to an achievable secrecy rate allocation  and belongs to the core.  
\subsection{Problem statement and game definition}  \label{secmod2}

\begin{figure} 
\centering   
 \includegraphics[width=8.5cm]{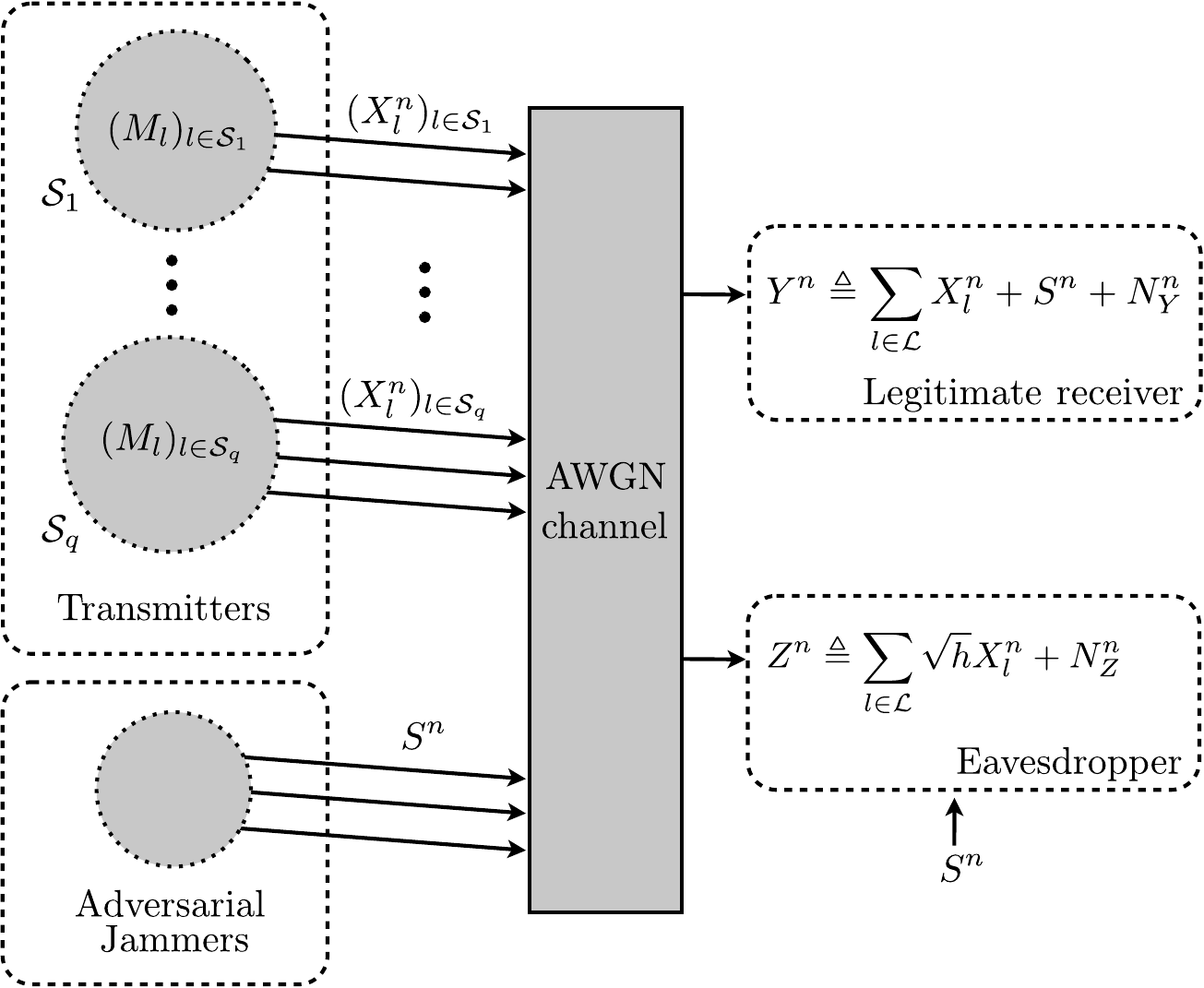}
  \caption{Degraded GMAC-WT with selfish transmitters in the presence of adversarial jammers, where the transmitters form $q$ coalitions. In the absence of adversarial jammers, set $S^n \leftarrow \emptyset$.} \label{fig:modelj}
\end{figure}
We consider the GMAC-WT-AJ with parameters $((\Gamma_l)_{l\in \mathcal{L}},h, \Lambda = 0, 1, 1)$, $h \in [0,h_{\Lambda}[$, i.e., a degraded GMAC-WT \cite{tekin2008gaussian}, i.e., the channel gains $(h_l)_{l \in \mathcal{L}}$ are all equal to $h \in [0,1[$.  
We assume that the  transmitters are selfish, i.e., they are solely interested in maximizing their own secrecy rate.
 The transmitters could potentially form \emph{coalitions} to achieve this goal as depicted in Figure \ref{fig:modelj} when $S^n \leftarrow \emptyset$, in the sense that subsets of agents can agree on a multiple access protocol before the actual information transmission to the receiver occurs. As would be the case for GMAC-WT, the members of a given coalition do not alter the multiple access protocol they agreed on once transmission
commences. 

\begin{rem} \label{rem2}
Choosing $\Lambda \neq 0$ offers a generalization of our result to the GMAC-WT with selfish transmitters in the presence of adversarial jammers, whose setting is depicted in Figure \ref{fig:modelj}. In the following, we derive all our results for $\Lambda \neq 0$ to enable this generalization.   
\end{rem}

The questions we would like to address are as follows. 
\begin{enumerate}[(i)]
\item Can the transmitters benefit from forming coalitions?  
\item If yes, can the transmitters find a consensus about which coalitions to form despite their selfishness? 
\item If such consensus exists, how should the secrecy sum-rate of each coalition be allocated among its transmitters? 
\end{enumerate} 
 
To answer these questions, we adopt a coalitional game theory framework,  e.g.,   \cite[Section~2.1]{peleg2007introduction}, by associating with each potential coalition of transmitters $\mathcal{S} \subseteq \mathcal{L}$ a worth $v(\mathcal{S})$. As detailed in Section \ref{sec:supadd}, such function $v$ will allow us to study stability of coalitions formed by the transmitters, where stability of a coalition means that there is no incentive to merge with another coalition or to split in smaller coalitions. This approach is similar to the approach taken in  \cite{la2004game} in the absence of security constraints. For completeness, we review the definition of the value function. 
To this end, we first define a game corresponding to our problem as follows. 
For $l\in \mathcal{L}$, let $\mathcal{A}_l$ corresponds to the set of all the possible strategies that transmitter $l$ can adopt, and let $\pi_l(a_{\mathcal{L}})$ be the payoff of transmitter $l$, i.e., its secrecy rate, when the strategies $a_{\mathcal{L}} \in \bigtimes_{l\in\mathcal{L}} \mathcal{A}_l$ are  played by all the transmitters. Two potential choices for the worth $v(\mathcal{S})$ of coalition $\mathcal{S} \subseteq \mathcal{L}$  are the following,\mbox{\cite{aumann1960neumann,jentzsch1964some}} 
\begin{align}
\max_{a_{\mathcal{S}}} \min_{a_{\mathcal{S}^c}} \sum_{i \in \mathcal{S}}\pi_i(a_{\mathcal{S}},a_{\mathcal{S}^c}), \label{eq6}\\
 \min_{a_{\mathcal{S}^c}} \max_{a_{\mathcal{S}}} \sum_{i \in \mathcal{S}}\pi_i(a_{\mathcal{S}},a_{\mathcal{S}^c}),\label{eq7}
\end{align}
where the quantity in \eqref{eq6} corresponds to the payoff that coalition $\mathcal{S}$ can ensure to its members regardless of the strategies adopted by the member of $\mathcal{S}^c$, and the one in \eqref{eq7} to the payoff that coalition $\mathcal{S}^c$ cannot prevent coalition $\mathcal{S}$ to receive. See, for instance, \cite{shapley1973gameb} for a detailed explanation of the subtle difference between these two notions in general.  Observe that, for our problem, both quantities are equal since for any $\mathcal{S} \subseteq \mathcal{L}$, there exists $a^*_{\mathcal{S}^c} \in \bigtimes_{i\in\mathcal{S}^c} A_i$  such that for any strategies $a_{\mathcal{S}}\in \bigtimes_{i\in\mathcal{S}} A_i$, we have $$ \sum_{i \in \mathcal{S}}\pi_i(a_{\mathcal{S}},a_{\mathcal{S}^c}) \geq \sum_{i \in \mathcal{S}}\pi_i(a_{\mathcal{S}},a^*_{\mathcal{S}^c}).$$
Indeed, note that the signals of the transmitters in $\mathcal{S}^c$ can be considered as a single signal of power up to $(\sum_{l \in\mathcal{S}^c} \sqrt{\Gamma_l})^2$ from the receiver perspective.  
We can thus consider the following strategy $a^*_{\mathcal{S}^c}$:  the transmitters in $\mathcal{S}^c$ collude against coalition~$\mathcal{S}$ by  acting as a mega jammer with power upper bounded by $(\sum_{l \in\mathcal{S}^c} \sqrt{\Gamma_l})^2$ and by revealing their transmitted sequences to the eavesdropper. Using the terminology of \cite{jentzsch1964some}, the game is \textit{clear}, i.e., equality holds between \eqref{eq6} and~\eqref{eq7}.

To summarize, we cast the problem as a coalitional game $(\mathcal{L},v)$ where the value function is defined as
\begin{align} \label{eqvf}
v  : 2^{\mathcal{L}}& \to \mathbb{R}^+, 
 \mathcal{S} \mapsto \max_{a_{\mathcal{S}}} \min_{a_{\mathcal{S}^c}} \sum_{i \in \mathcal{S}}\pi_i(a_{\mathcal{S}},a_{\mathcal{S}^c}),
\end{align}
 such that  $v(\mathcal{S})$ corresponds to the maximal secrecy sum-rate achievable by coalition $\mathcal{S}$ when \emph{no specific strategy is assumed} for the transmitters in $\mathcal{S}^c$, who act as jammers. Note that the definition of the game is identical to the one in \cite{la2004game}, except that the value function now needs to account for our security constraints. The characterization of this value function is provided in Theorem \ref{th1}.
\subsection{Properties of the game and characterization of its core} \label{sec:supadd}

We first show the following characterization of the value function defined in \eqref{eqvf}, which is a consequence of Theorem~\ref{thsumrate}. 
\begin{thm} \label{th1}
Let $\mathcal{S} \subseteq \mathcal{L}$. We have
	\begin{equation}
		v(\mathcal{S})= \left[ \frac{1}{2} \log \left( \frac {{1+ h_{\Lambda_{\mathcal{S}^c}}\Gamma_{\mathcal{S}(\Lambda_{\mathcal{S}^c})}   }} { {1+  h\Gamma_{\mathcal{S}(\Lambda_{\mathcal{S}^c})} }}\right)   \right]^+,
	\end{equation}
	where for any $\mathcal{S} \subseteq \mathcal{L}$, $\Lambda_{\mathcal{S}^c} \triangleq \left( \sqrt{\Lambda} + \sum_{l\in\mathcal{S}^c} \sqrt{\Gamma_l} \right)^2$, $h_{\Lambda_{\mathcal{S}^c}} \triangleq (1+\Lambda_{\mathcal{S}^c})^{-1}$, $\mathcal{S}(\Lambda_{\mathcal{S}^c}) \triangleq \{ l\in \mathcal{S} : \Gamma_l > \Lambda_{\mathcal{S}^c} \}$, $\Gamma_{\mathcal{S}} \triangleq \sum_{l \in \mathcal{S}} \Gamma_l$.
\end{thm}
Observe that, as expected, the characterization of the value function in Theorem \ref{th1} recovers the one in \cite{la2004game} by setting $h=0$. We now review the notion of superadditivity.
\begin{defn}
A game $(\mathcal{L},v)$ is superadditive if $v: 2^{\mathcal{L}} \to \mathbb{R}^+$ is such that 
\begin{align}
\forall \mathcal{S},\mathcal{T}\subseteq \mathcal{L}, \mathcal{S} \cap\mathcal{T} = \emptyset \implies  v(\mathcal{S}) + v(\mathcal{T}) \leq v(\mathcal{S} \cup \mathcal{T}) .
\end{align}
\end{defn}

\begin{pro}
The game $(\mathcal{L},v)$ defined in \eqref{eqvf} is superadditive.
\end{pro}
\begin{proof}
 The secrecy constraints for coalitions  $\mathcal{S}$ and $\mathcal{T}$, with $\mathcal{S} \cap\mathcal{T} = \emptyset$, implies a secrecy constraint for the coalition $\mathcal{S} \cup \mathcal{T}$:
\begin{subequations}
\begin{align}
&I\left(M_{\mathcal{S}} M_{\mathcal{T}} ;  Z^{kn} X^{kn}_{(\mathcal{S} \cup \mathcal{T})^c}\right) \\
& = I\left(M_{\mathcal{S}} ;  Z^{kn} X^{kn}_{(\mathcal{S} \cup \mathcal{T})^c}\right) + I\left(M_{\mathcal{T}}  ;  Z^{kn} X^{kn}_{(\mathcal{S} \cup \mathcal{T})^c} |  M_{\mathcal{S}}\right)\\
& \leq I\left(M_{\mathcal{S}} ;  Z^{kn} X^{kn}_{(\mathcal{S} \cup \mathcal{T})^c}\right) + I\left(M_{\mathcal{T}}  ;  Z^{kn} X^{kn}_{(\mathcal{S} \cup \mathcal{T})^c}   X^{kn}_{\mathcal{S}}M_{\mathcal{S}}\right) \label{eq20a}\\
& \leq I\left(M_{\mathcal{S}} ;  Z^{kn} X^{kn}_{(\mathcal{S} \cup \mathcal{T})^c}\right) + I\left(M_{\mathcal{T}}  ;  Z^{kn} X^{kn}_{(\mathcal{S} \cup \mathcal{T})^c} X^{kn}_{\mathcal{S}}\right) + \epsilon'_n \label{eq20aa} \\
& \leq I\left(M_{\mathcal{S}} ;  Z^{kn} X^{kn}_{\mathcal{S}^c}\right) + I\left(M_{\mathcal{T}}  ;  Z^{kn} X^{kn}_{\mathcal{T}^c} \right) + \epsilon'_n, \label{eq20b}
\end{align}
\end{subequations}
where  \eqref{eq20aa} holds by Fano's inequality assuming that for any $n \in \mathbb{N}$, $\epsilon_n' \triangleq nk\epsilon_n  + 1$ and $\mathbb{P}[\widehat{M}_{\mathcal{S}} \neq M_{\mathcal{S}}] \leq  \epsilon_n$ with $ \epsilon_n \xrightarrow{n\to \infty} 0$, \eqref{eq20b} holds because $\mathcal{S} \cap\mathcal{T} = \emptyset$.
\end{proof}

Superadditivity implies that there is an interest in forming a large coalition to obtain a larger secrecy sum-rate, however, large coalition might not be in the individual interest of the transmitters, and can thus be unstable. A useful concept to overcome this complication is the core of the game.

\begin{defn}[e.g. \cite{maschler1979geometric}] \label{defcore}
The core of a superadditive game $(\mathcal{L},v)$ is defined as follows.
\begin{align} 
&\mathcal{C}(v) \triangleq \nonumber \\
& \left\{ (R_l)_{l \in \mathcal{L}} : \sum_{l \in \mathcal{L}} R_l = v(\mathcal{L}) \text{ and }\sum_{i \in \mathcal{S}} R_i \geq v(\mathcal{S}), \forall \mathcal{S} \subset \mathcal{L} \right\}. \label{eqcore}
\end{align}
\end{defn}

Observe that for any point in the core, the grand coalition, i.e., the coalition $\mathcal{L}$, is in the best interest to all transmitters, since the set of inequality in \eqref{eqcore} ensures that no coalition of agents can increase its secrecy sum-rate by leaving the grand coalition. Observe also that for any point in the core the maximal secrecy sum rate $v(\mathcal{L})$ for the grand coalition  is achieved. 
In general, the core of a game can be empty \cite{myerson2013game}. However, we will show that the game $(\mathcal{L},v)$ defined in~\eqref{eqvf} has a non-empty core. 

Definition \ref{defcore} further clarifies the choice of the value function $v$. A coalition $\mathcal{S}$ wishes to be associated with a value $v(S)$ as large as possible, while the transmitters outside $\mathcal{S}$ wish $v(S)$ to be as small as possible to demand a higher share of $v(\mathcal{L})$. The latter transmitters achieve their goal by waiving a threat argument, which consists in arguing that they could adopt the strategy that minimizes $v(S)$, whereas coalition $\mathcal{S}$ achieves its goal by arguing that it can always achieve the secrecy sum-rate  of Theorem~\ref{thsumrate}, irrespective of the strategy of transmitters in $\mathcal{S}^c$. This formulation is generically termed as alpha effectiveness or alpha theory \cite{aumann1960neumann,jentzsch1964some,shapley1973gameb}.  It has also been used in~\cite{la2004game} for the Gaussian multiple access channel and in \cite{ISIT172} for secret-key generation in many-to-one networks. 

\begin{rem} \label{remcore}
The core can also be understood as a converse for our problem since it provides upper bounds for $R_{\mathcal{S}}$, $\mathcal{S}\subseteq \mathcal{L}$. More specifically, we have the following alternative characterization of the core of the game $(\mathcal{L},v)$. 
\begin{align} 
&\mathcal{C}(v) = \bigg\{ (R_l)_{l \in \mathcal{L}} : \forall \mathcal{S} \subseteq \mathcal{L}, \nonumber \\
&  \phantom{-} \left[ \frac{1}{2} \log \left( \frac{1 + h_{\Lambda_{\mathcal{S}^c}} \Gamma_{\mathcal{S}(\Lambda_{\mathcal{S}^c})} } {1 + h \Gamma_{\mathcal{S}(\Lambda_{\mathcal{S}^c})}  } \right)  \right]^+ \leq  R_{\mathcal{S}}  \nonumber \\
& \left. \phantom{-}\leq \frac{1}{2} \!\log\! \left(\! \frac{1+ h_{\Lambda} \Gamma_{\mathcal{L}(\Lambda)}}{1+ h \Gamma_{\mathcal{L}(\Lambda)}} \!\! \right)   
 \!-\! \left[ \frac{1}{2}\! \log \!\left(\! \frac{1 + h_{\Lambda_{\mathcal{S}}} \Gamma_{\mathcal{S}^c(\Lambda_{\mathcal{S}})} }{1 + h \Gamma_{\mathcal{S}^c(\Lambda_{\mathcal{S}})}  } \right)  \!  \right]^+ \!\! \right\}\!. \label{eqcorealt}
\end{align}
\end{rem} 
\begin{proof}
We obtain \eqref{eqcorealt} from the following equivalences
\begin{subequations}
\begin{align}
& \left(\sum_{l \in \mathcal{L}} R_l = v(\mathcal{L}) \text{ and }\sum_{i \in \mathcal{S}} R_i \geq v(\mathcal{S}), \forall \mathcal{S} \subset \mathcal{L} \right)\\
&\iff  \left( \sum_{i \in \mathcal{S}} R_i = v(\mathcal{L}) -  \sum_{i \in \mathcal{S}^c} R_i \right. \nonumber\\
& \phantom{-----} \left. \text{ and }\sum_{i \in \mathcal{S}} R_i \geq v(\mathcal{S}), \forall \mathcal{S} \subset \mathcal{L} \right)\\
& \iff \left( v(\mathcal{L}) -  v(\mathcal{S}^c) \geq \sum_{i \in \mathcal{S}} R_i \geq v(\mathcal{S}), \forall \mathcal{S} \subseteq \mathcal{L} \right).
\end{align}
\end{subequations}
\end{proof}
Note that the core may contain rates that are not achievable. We next characterize a subset of the core that is achievable by the transmitters, i.e., that is included in $\mathcal{R}$ defined in Theorem~\ref{thregion} for the GMAC-WT-AJ with parameters  $((\Gamma_l)_{l\in \mathcal{L}},h, \Lambda , 1, 1)$. 
 This characterization will later be useful to prove that the stable and fair allocation found in the next section is achievable and belongs to the core.

\begin{thm} \label{propcore}
$(i)$ The core of the game $(\mathcal{L},v)$ contains the following  rate-tuples
\begin{align} 
\mathcal{C}^*(v) &\triangleq \bigg\{ (R_l)_{l \in \mathcal{L}} :  \forall l \in \mathcal{L}^c(\Lambda), R_l = 0,  \nonumber \\
& \phantom{---}R_{\mathcal{L}(\Lambda)}  =  \frac{1}{2} \log \left( \frac{1+ h_{\Lambda} \Gamma_{\mathcal{L}(\Lambda)}}{1+ h \Gamma_{\mathcal{L}(\Lambda)}}  \right), \text{ and }   \nonumber \\
& \phantom{---} \forall \mathcal{S} \subset \mathcal{L}(\Lambda), \left. R_{\mathcal{S}}   \leq  \frac{1}{2} \log \left( \frac{1+ h_{\Lambda} \Gamma_{\mathcal{S}}}{1+ h \Gamma_{\mathcal{S}}}  \right)   
    \right\}. \label{eqrates}
\end{align}
$(ii)$ The rate-tuples in $\mathcal{C}^*(v)$ are achievable.
\end{thm}
\begin{proof}
See Appendix \ref{App_th4}.
\end{proof}

\subsection{Stable and fair allocation} \label{sec:alloc}
Although we have found in Theorem \ref{propcore} achievable rate-tuples that belong to the core of the game, the question of choosing a specific point in the core remains. We use the solution concept introduced in \cite{la2004game} to define a fair secrecy rate allocation. We will then show  (i) existence and uniqueness, (ii) achievability, and (iii) belonging to the core of this allocation. 
\begin{defn}[\cite{la2004game}] \label{defax}
A fair secrecy rate allocation $\{R^*_l(v)\}_{l\in\mathcal{L}}$ should satisfy the following axioms.
 \begin{enumerate}[(i)]
 \item \textbf{Efficiency}: The maximal secrecy sum rate is achieved $\sum_{l\in \mathcal{L}} R^*_l(v) =  v(\mathcal{L})$.
 \item \textbf{Symmetry}: The labeling of the players should not influence the secrecy rate allocation. More specifically, let $\pi \in \textup{Sym}(L)$, where $\textup{Sym}(L)$ is the symmetric group on $\mathcal{L}$, and let $\pi v$ be the game with value function that maps $\mathcal{S} \subseteq \mathcal{L}$ to $v(\{ \pi(s) : s\in \mathcal{S} \})$. Then, for any $\pi \in \textup{Sym}(L)$, for any $l\in \mathcal{L}$, $R^*_{l}(v) = R^*_{\pi(l)}(\pi v)$.
 \item \textbf{Envy-freeness}: For $i,j\in\mathcal{L}$, if $\Gamma_i> \Gamma_j$ and player $i$ decides to conserve energy and only use the power $\Gamma_j$, then player $i$ should receive the same secrecy rate allocation than player~$j$. More specifically, let $v^{i,j}$ be the same game as $v$ when the power constraint of player $i$ is $\Gamma_j$, then one should have $R^*_i(v^{i,j}) = R^*_j(v)$. This axiom is based on the notion of envy~\cite{feldman1974fairness}.
\end{enumerate}  
\end{defn}
In Proposition \ref{prop1}, we  show that for our problem there is a unique secrecy rate allocation as axiomatized in Definition \ref{defax}.

\begin{prop} \label{prop1}
There exists a unique secrecy rate allocation $\{R^*_l(v)\}_{l\in\mathcal{L}}$ that satisfies the three axioms Efficiency, Symmetry, and Envy-freeness of Definition \ref{defax}. Moreover,
\begin{subequations} \label{eqoptalloc}
\begin{align} 
&\forall l \in \mathcal{L}^c(\Lambda),  R^*_{l}(v) = 0,\\
&\forall l \in \mathcal{L}(\Lambda), \nonumber\\&  R^*_{l}(v) = \frac{1}{l} \left[\frac{1}{2} \log \left[\frac{ 1+ h_{\Lambda} (l\overline{\Gamma}_l + \overline{\Gamma}_{l+1:L})}{1+ h(l\overline{\Gamma}_l + \overline{\Gamma}_{l+1:L}) } \right] - R^*_{l+1:L}(v)\right],
\end{align}
\end{subequations}
where we have defined for $l \in \mathcal{L}$, 
\begin{equation}
\overline{\Gamma}_l \triangleq \begin{cases}
\Gamma_l, \text{ if } l\in \mathcal{L}(\Lambda) \\
0, \text{ if } l\in \mathcal{L}^c(\Lambda) \\
\end{cases},
\end{equation}
and we have used the notation $\overline{\Gamma}_{l+1:L} \triangleq \sum_{i=l+1}^L \overline{\Gamma}_l$, and $R^*_{l+1:L}(v)\triangleq \sum_{j=l+1}^L R^*_{j}(v)$.
\end{prop}

\begin{proof}
See Appendix \ref{App_prop1}.
\end{proof}

In Theorem \ref{th3}, we prove that the secrecy rate allocation $\{R^*_l(v)\}_{l\in\mathcal{L}}$ from Definition \ref{defax} is achievable and belongs to the core. Note that in the absence of security constraints, i.e., $h=0$, our results recovers \cite{la2004game}.

\begin{thm} \label{th3}
$(i)$ The secrecy rate allocation $\{R^*_l(v)\}_{l\in\mathcal{L}}$ defined in \eqref{eqoptalloc} is such that
\begin{align}
\forall \mathcal{S} \subseteq \mathcal{L}(\Lambda),  0\leq R^*_{\mathcal{S}}(v)   \leq \frac{1}{2} \log \left( \frac{1+  h_{\Lambda} \Gamma_{\mathcal{S}}}{1+ h \Gamma_{\mathcal{S}}}  \right). \label{equation17}
\end{align}
$(ii)$ By $(i)$, $\{R^*_l(v)\}_{l\in\mathcal{L}}$ is in $\mathcal{C}^*(v)$ and is thus achievable by Theorem \ref{propcore}.\\
$(iii)$ By $(ii)$, $\{R^*_l(v)\}_{l\in\mathcal{L}}$ belongs to the core because $\mathcal{C}^*(v) \subseteq \mathcal{C}(v)$ from Theorem \ref{propcore}. 
\end{thm}
\begin{proof}
See Appendix \ref{App_th3}.
\end{proof}
In Proposition \ref{prop2}, we bound the ratio of the secrecy rates of two transmitters by studying the influence of the noise level at the legitimate receiver and at the eavesdropper.  More specifically, for $\omega >0$, we make the following substitution in our model, i.e., Equation \eqref{eqmod1}, $ \sigma^2_Y \leftarrow  \omega \sigma^2_Y $ and $ \sigma^2_Z \leftarrow  \omega \sigma^2_Z$, $\omega \in \mathbb{R}_+$ such that after normalization for any $l\in \mathcal{L}$, $\Gamma_l \leftarrow \omega^{-1} \Gamma_l$ and $\Lambda \leftarrow \omega^{-1} \Lambda$.  Let $v^{(\omega)}$ denote the game with these new parameters.    
\begin{prop} \label{prop2}
We assume the sequence $(\Gamma_{l})_{l\in \mathcal{L}}$ decreasing by relabeling the players if necessary. Define $L(\Lambda) \triangleq |\mathcal{L}(\Lambda)|$. For any $l \in \llbracket 1, L(\Lambda)-1 \rrbracket$ such that $\Gamma_{l} \neq \Gamma_{l+1}$, $\omega \mapsto \frac{R_{l}^*(v^{(\omega)}) }{R_{l+1}^*(v^{(\omega)}) }$ from $\mathbb{R}_+^*$ to $\mathbb{R}$ is increasing and its image is $\left[1, \frac{\Gamma_{l'}}{\Gamma_l} \right]$. Hence, the secrecy rate allocation $\{R^*_l(v)\}_{l\in\mathcal{L}}$ defined in~\eqref{eqoptalloc} satisfies for $l,l'\in\mathcal{L}(\Lambda)$ such that $\Gamma_{l'} \geq \Gamma_l$
\begin{align}
1 \leq \frac{R^*_{l'}(v)}{R^*_l(v)} \leq \frac{\Gamma_{l'}}{\Gamma_l}.
\end{align}
\end{prop}

\begin{proof}
See Appendix \ref{App_prop2}.
\end{proof}
Proposition \ref{prop2} displays the same qualitative property as in \cite[Section 5.5]{la2004game}, namely, when the signal-to-noise ratio is high for all transmitters, they all obtain similar secrecy rates, whereas when the signal-to-noise ratio is low, they obtain secrecy rates proportional to their power constraints.

We illustrate Theorem \ref{propcore}, Proposition \ref{prop1}, and Theorem \ref{th3} in Figure \ref{fig:ex} with an example when $L=2$ and $\Lambda =0$. From Figure \ref{fig:ex}, as stated in Theorem \ref{propcore}, we observe that all the rates in $\mathcal{C}^*(v)$ are achievable and in the core $\mathcal{C}(v)$, however, this inclusion is strict, in general. From Figure \ref{fig:ex}, we also see that the core, characterized in Remark \ref{remcore}, may contain rates that are not achievable. However, the unique fair allocation (that satisfies Definition \ref{defax}) and is characterized in Proposition \ref{prop1} can be seen to belong to $\mathcal{C}^*(v)$ in Figure \ref{fig:ex}, and is thus achievable and in the core, which is formally proved in Theorem \ref{th3}.

\begin{figure} 
\centering   
  \includegraphics[width=8.5cm]{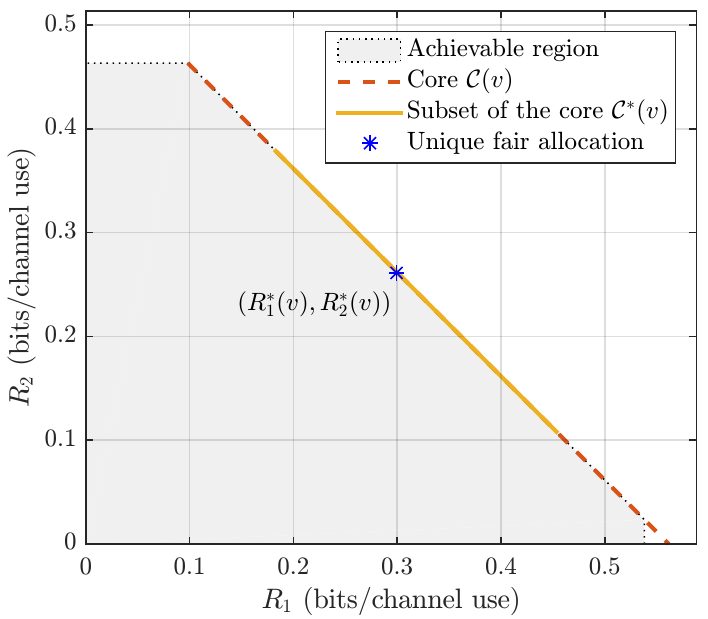}
  \caption{Representation of a known achievable region \cite{tekin2008gaussian} for the degraded GMAC-WT, the core $\mathcal{C}(v)$ defined in Remark \ref{remcore}, a subset of the core $\mathcal{C}^*(v)$ defined in Theorem \ref{propcore}, and the allocation $(R_1^*(v),R_2^*(v))$ defined in Proposition \ref{prop1} for two transmitters with power constraints $(\Gamma_1,\Gamma_2)= (2,1.4)$ with $h=0.3$ and $\Lambda =0$.} \label{fig:ex}
\end{figure}

\section{Non-degraded GMAC-WT with two selfish transmitters} \label{sec:game2}
We define in Section \ref{secmod2n} a coalitional game for the non-degraded  GMAC-WT with two selfish transmitters. The case of more than two transmitters remains an open problem as explained in Section \ref{secmod2n}. We study the properties of the game and its core in Sections \ref{sec:supaddn}, \ref{sec:allocn}. In particular, we demonstrate in Section \ref{sec:supaddn} that unlike the degraded case in Section \ref{sec:game}, cooperation might not always be in the best interest of the transmitters. Then, we identify in Section \ref{sec:allocn} sufficient conditions to have a setting in which cooperation is beneficial to both transmitters. In such settings, we also propose as solution concept for a fair allocation, an achievable secrecy rate allocation that is in the core and that satisfies a series of axioms that generalizes Section~\ref{sec:alloc}.

\subsection{Problem statement and game definition} \label{secmod2n}
In this section, we consider the two-transmitter GMAC-WT-AJ with parameters $( (\Gamma_1, \Gamma_2), (h_1,h_2), \Lambda = 0, \sigma^2_Y, \sigma^2_Z)$,  i.e., a general GMAC-WT \cite{tekin2008general} with two transmitters. Additionally, we assume that the  transmitters are selfish.
\begin{rem}
Similar to Remark \ref{rem2}, choosing $\Lambda \neq 0$ offers a generalization of our result to the non-degraded GMAC-WT with two selfish transmitters in the presence of adversarial jammers. In the following, we derive all our results for $\Lambda \neq 0$ to enable this generalization.  
\end{rem}

While one can adopt the same formalism as in Section \ref{sec:model} to define a coalitional game, the non-degraded GMAC-WT with selfish transmitters is challenging to deal with because a characterization of the right hand side in~\eqref{eqvf} is unknown for arbitrary $L$. Even for the case $L=2$, a characterization of $v(\{1,2\})$ is unknown. 

We propose to define a game for $L=2$ as follows. One can choose $v(\{1,2\})$ as the best known achievable sum-rate by two transmitters, denoted by $R_{1,2}^*$, and still requiring for one transmitter not to make any assumption on the communication strategy of the other. Indeed, (i)~the communication strategy to achieve $R_{1,2}^*$ must be known by both transmitters to be implemented (if only one transmitter is aware of a strategy that achieves a better sum-rate, then the strategy cannot be implemented), and (ii) $\max_{a_{\mathcal{S}}} \min_{a_{\mathcal{S}^c}} \sum_{i \in \mathcal{S}}\pi_i(a_{\mathcal{S}},a_{\mathcal{S}^c})$ can be determined for $\mathcal{S} = \{ 1\}$ and $\mathcal{S} = \{ 2\}$ from Theorem \ref{thsumrate}. Hence, similar to Section \ref{sec:model}, we define a coalitional game to model a non-degraded multiple access wiretap channel with two selfish transmitters (potentially in the presence of adversarial jammers) by
\begin{subequations} \label{eq_game2}
\begin{align}
v(\{1\}) & \triangleq \max_{a_{\{1\}}} \min_{a_{\{2\}}} \sum_{i \in \mathcal{S}}\pi_i(a_{\{1\}},a_{\{2\}}),\\ 
v(\{2\}) & \triangleq  \max_{a_{\{2\}}} \min_{a_{\{1\}}} \sum_{i \in \mathcal{S}}\pi_i(a_{\{2\}},a_{\{1\}}),\\ 
v(\{1,2\}) & \triangleq R_{1,2}^*.
\end{align}
\end{subequations}

\subsection{Cooperation might not be beneficial to all transmitters} \label{sec:supaddn}

We first give the following characterization of the value function for the game defined in \eqref{eq_game2}, which follows from~\cite{ISIT17}, where it has been shown that $R_{1,2}^* = \left[\frac{1}{2} \log \left(\frac{1 + h_{\Lambda}(\Gamma_1 +  \Gamma_2)}{1+ h_1\Gamma_1 + h_2 \Gamma_2} \right)\right]^+ $. Theorem~\ref{theorem6ref} is the counterpart of Theorem~\ref{th1} for the degraded case.
\begin{thm} \label{theorem6ref}
The value function $v$ of the game defined in Section \ref{secmod2n} can be characterized as follows.
\begin{subequations} \label{eqgame2}
\begin{align}
v(\{1\}) & = \mathds{1} \left\{ \Gamma_1 >(\sqrt{\Gamma_2}+\sqrt{\Lambda})^2 \right\} \nonumber \\
& \phantom{---}\times \left[ \frac{1}{2} \log \left(\frac{1 +  \frac{\Gamma_1}{1+ (\sqrt{\Gamma_2}+\sqrt{\Lambda})^2} }{1+ h_1\Gamma_1} \right) \right]^+,\\ 
v(\{2\}) & =  \mathds{1} \left\{ \Gamma_2 >(\sqrt{\Gamma_1}+\sqrt{\Lambda})^2 \right\} \nonumber \\
& \phantom{---}\times\left[ \frac{1}{2} \log \left(\frac{1 +  \frac{\Gamma_2}{1+ (\sqrt{\Gamma_1}+\sqrt{\Lambda})^2} }{1+ h_2\Gamma_2} \right) \right]^+,\\ 
v(\{1,2\}) & = \left[\frac{1}{2} \log \left(\frac{1 + h_{\Lambda}(\Gamma_1 +  \Gamma_2)}{1+ h_1\Gamma_1 + h_2 \Gamma_2} \right) \right]^+.
\end{align}
\end{subequations}
\end{thm}

In Proposition \ref{propcounter}, stated next, we identify a range of parameter values for which the grand coalition might not be in the best interest of both transmitters.
\begin{prop} \label{propcounter}
When $h_1 \notin [0,h_{\Lambda}[$ or $h_2 \notin [0,h_{\Lambda}[$, the grand coalition might not form.
\end{prop}

\begin{proof}
It is sufficient to exhibit an example for which $h_1 \notin [0,h_{\Lambda}[$ or $h_2 \notin [0,h_{\Lambda}[$ and $v(\{ 1,2\}) < v(\{1\})$. We set 
$h_1 = 0.1$, $h_2 = 1.5$, $\Gamma_1 = 1$, $\Gamma_2= 0.4$, and $\Lambda =0.1$. Observe that $\Gamma_1 = 1 > 0.9 = (\sqrt{\Gamma_2}+\sqrt{\Lambda})^2$.
We numerically obtain
\begin{align}
v(\{ 1,2\}) 
= \frac{1}{2} \log \left(\frac{1 + h_{\Lambda}(\Gamma_1 +  \Gamma_2)}{1+ h_1\Gamma_1 + h_2 \Gamma_2} \right)
 < 0.2095,
\end{align}
\begin{align}
 0.2362 
 < \frac{1}{2} \log \left( \frac{1 +  \frac{\Gamma_1}{1+ (\sqrt{\Gamma_2}+\sqrt{\Lambda})^2} }{1+ h_1\Gamma_1} \right)
 = v(\{1\}).
\end{align}

Hence, $v(\{ 1,2\}) < v(\{1\})$, and Transmitter $1$ has no interest in engaging in a collective protocol with Transmitter~2.
\end{proof}
Note that for any parameters $(h_1,h_2, \Gamma_1, \Gamma_2, \Lambda)$ such that $v(\{ 1,2\}) < v(\{1\})+v(\{2\})$, cooperation is not in the best interest of the transmitters and the core of the game is empty. It is possible to refine Proposition \ref{propcounter} and provide in Proposition~\ref{proposition4} a necessary and sufficient condition on the model parameters for the core of the game to be empty. 
\begin{prop} \label{proposition4}
For $l \in \{1,2\}$, define $\bar{l} \triangleq 3-l.$ Assume $\Lambda = 0$ and $v(\{1,2\})>0$. Cooperation is not in the best interest of both transmitters and the core of the game is empty if and only if
\begin{align}
\max_{l\in \{ 1,2 \}}  	\min \left[ \frac{\Gamma_l}{\Gamma_{\bar{l}}} ,    \frac{1+ h_1\Gamma_1 + h_2 \Gamma_2 }{(1+ h_l\Gamma_l)(1+ \Gamma_{\bar{l}})} \right]  > 1. 
 \label{eqprop4}
\end{align}
\end{prop}
\begin{proof}
Observe that when $\Lambda = 0$, $v(\{1\}) + v(\{2\}) = \max ( v(\{1\}) , v(\{2\}))$.	\eqref{eqprop4} translates the condition $v(\{ 1,2\}) < \max ( v(\{1\}) , v(\{2\}))$ using \eqref{eqgame2} and the fact that $v(\{1,2\})>0 \Leftrightarrow  \Gamma_1 +  \Gamma_2 >  h_1\Gamma_1 + h_2 \Gamma_2$.
\end{proof}

As a numerical example, we set $\Gamma_1 = 1$, $\Gamma_2= 0.4$, and $\Lambda =0.1$. We then vary $h_1$ and $h_2$ between $0$ and $2$ with a step size of $0.1$, and plot in Figure \ref{fig:ex2} a red cross when the pair $(h_1,h_2)$ implies that cooperation between the two transmitters is beneficial. However, it is an open problem to provide a general closed-form formula that depends on the channel parameters $h_1$ and $h_2$ to determine whether cooperation is beneficial to the transmitters.

\begin{figure} 
\centering   
  \includegraphics[width=8.5cm]{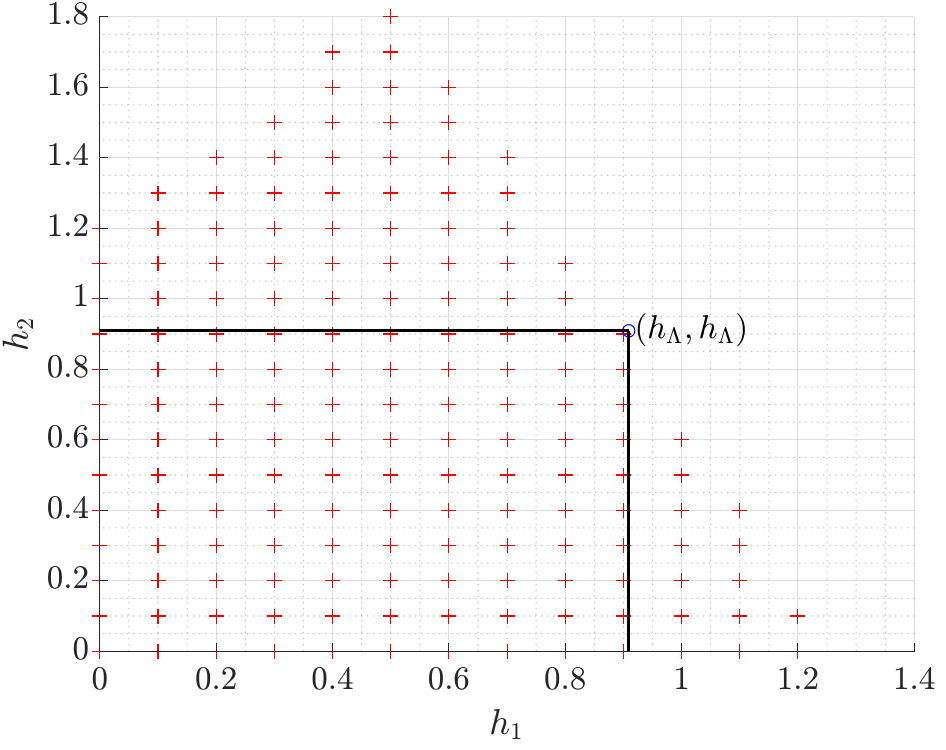}
  \caption{Fix $\Gamma_1 = 1$, $\Gamma_2= 0.4$, and $\Lambda =0.1$. $h_1$ and $h_2$ vary between $0$ and $2$ with a step size of $0.1$. A red cross for the pair $(h_1,h_2)$ indicates that cooperation between the two transmitters is beneficial. } \label{fig:ex2}
\end{figure}

\subsection{Stable and fair allocation when cooperation is in the best interest of all} \label{sec:allocn}

In the following, we focus on the case $h_1 \in [0,h_{\Lambda}[$, $h_2 \in [0,h_{\Lambda}[$ and show that it is a sufficient condition for cooperation to be in the best interest of both transmitters. Note that, in general, additional values for $h_1$ and $h_2$ could lead to cooperation being in the best interest of the transmitters, as illustrated in Figure \ref{fig:ex2}. We also assume that $\min(\Gamma_1 ,\Gamma_2) > \Lambda$, the problem is trivial otherwise. We further propose as fair solution concept, a unique achievable rate in the core that satisfies a series of axioms. Note that Definition \ref{defax} is not well-defined for a non-degraded GMAC-WT-AJ. Indeed, the Envy-freeness axiom cannot be only defined through the power constraints of the transmitters. More specifically, the configuration associated with a transmitter does not only depend on its power constraint but also on its channel gain with respect to the eavesdropper. We thus propose the following definitions.

\begin{defn}
For any $i$, for any $\Gamma^*\in ]0, \Gamma_i]$, for any let  $v_{\Gamma_i \leftarrow \Gamma^*, h_i \leftarrow h^*}$ be the same game as $v$ where the power constraint $\Gamma_i$ of Transmitter $i$ is replaced by $\Gamma^*$, and the channel gain associated with Transmitter $i$ at the eavesdropper is replaced by $h^*$.
\end{defn}

\begin{defn} \label{defax2}
A fair secrecy rate allocation $\{R^*_l(v)\}_{l\in\mathcal{L}}$ should satisfy the following axioms.
 \begin{enumerate}[(i)]
 \item \textbf{Efficiency}: Defined as in Definition \ref{defax}.
 \item \textbf{Symmetry}: Defined as in Definition \ref{defax}.
 \item \textbf{Envy-freeness}: For any $i,j \in \mathcal{L}$, if there exists $\Gamma_{i}^* \in ]0, \Gamma_i]$ such that $v_{\Gamma_i \leftarrow \Gamma_{i}^*, h_i \leftarrow h_i}( \mathcal{L} ) = v_{\Gamma_i \leftarrow \Gamma_j, h_i \leftarrow h_j}( \mathcal{L} )$, then $$R^*_i(v_{\Gamma_i \leftarrow \Gamma_{i}^*, h_i \leftarrow h_i}) = R^*_j(v).$$
 The meaning of this envy-freeness axiom is the following. If Transmitter $i$ decides to conserve power by transmitting with some power $\Gamma_{i}^* \in ]0, \Gamma_i]$, and if this $\Gamma_{i}^*$ makes Transmitter~$i$ and Transmitter~$j$  contribute equally to the grand coalition value, i.e., $v_{\Gamma_i \leftarrow \Gamma_{i}^*, h_i \leftarrow h_i}( \mathcal{L} ) = v_{\Gamma_i \leftarrow \Gamma_j, h_i \leftarrow h_j}( \mathcal{L} )$, then Transmitter $i$ receives the same payoff as Transmitter $j$.
\end{enumerate}  
\end{defn}

Note that our new definition of envy-freeness recovers as special case the one in Definition~\ref{defax} for degraded channels, i.e., when $h_1=h_2$. However, it is a priori unclear whether or not this generalization will still lead to a unique solution that is in the core and achievable. We first show in Lemma \ref{lemstar} that either $\Gamma_1^*$ or $\Gamma_2^*$ exists, where $\Gamma_1^*$, $\Gamma_2^*$ are as in Definition \ref{defax2}.
\begin{lem} \label{lemstar}
Suppose that $h_1 \in [0,h_{\Lambda}[$, $h_2 \in [0,h_{\Lambda}[$. There exists $\Gamma_{1}^* \in ]0, \Gamma_1]$ such that $v_{\Gamma_1 \leftarrow \Gamma_{1}^*, h_1 \leftarrow h_1}( \mathcal{L} ) = v_{\Gamma_1 \leftarrow \Gamma_2, h_1 \leftarrow h_2}( \mathcal{L} )$, or there exists $\Gamma_{2}^* \in ]0, \Gamma_2]$ such that $v_{\Gamma_2 \leftarrow \Gamma_{2}^*, h_2 \leftarrow h_2}( \mathcal{L} ) = v_{\Gamma_2 \leftarrow \Gamma_1, h_2 \leftarrow h_1}( \mathcal{L} )$. Moreover, if $(\Gamma_{1}^*,\Gamma_{2}^*) \neq (\Gamma_{1},\Gamma_{2})$, then $\Gamma_{1}^* \in ]0, \Gamma_1]$ or (mutually exclusive or)  $\Gamma_{2}^* \in ]0, \Gamma_2]$.
\end{lem}

\begin{proof}
See Appendix \ref{App_lemstar}.
\end{proof}
Then, in Proposition \ref{Propaxi}, we show that there exists a unique secrecy rate allocation as axiomatized in Definition~\ref{defax2}, and that this rate allocation is achievable, and belongs to the core.
\begin{prop} \label{Propaxi}
Suppose that $h_1 \in [0,h_{\Lambda}[$, $h_2 \in [0,h_{\Lambda}[$. Assume that $\Gamma_{1}^* \in ]0, \Gamma_1]$ (Exchange the role of the two transmitters if $\Gamma_{2}^* \in ]0, \Gamma_2]$). The unique solution that satisfies the three axioms of Definition \ref{defax2} is
\begin{subequations} \label{eqalloc2n}
\begin{align}
R_2^* (v) & = \frac{1}{2} \left[ \frac{1}{2} \log \left( \frac{1 + 2  h_{\Lambda}\Gamma_2}{1 + 2 h_2 \Gamma_2} \right) \right], \\
R_1^* (v) & =  \frac{1}{2} \log \left( \frac{1 +   h_{\Lambda}(\Gamma_1+ \Gamma_2)}{1  + h_1 \Gamma_1 + h_2 \Gamma_2} \right) -   R_2^* (v).
\end{align}
\end{subequations}
\end{prop}
\begin{proof}
 One immediately sees that the allocation described in \eqref{eqalloc2n} satisfies the efficiency and symmetry axioms. One can also see that the Envy-freeness axiom of Definition \ref{defax2} is satisfied using \eqref{eqfndg} derived in the proof of Lemma \ref{lemstar}. Finally, unicity is proved similar to the proof of Proposition~\ref{prop1}.
\end{proof}

Next, assume $\Gamma_1 ,\Gamma_2 > \Lambda$. Recall that an achievable region for the two-transmitter GMAC-WT-AJ with parameters $( (\Gamma_1, \Gamma_2), (h_1,h_2), \Lambda, \sigma^2_Y, \sigma^2_Z)$
 is given by \cite{ISIT17}
\begin{align}
&\mathcal{R}_a \triangleq \nonumber \\
&\left\{ (R_1, R_2): R_1  \leq \left[ \frac{1}{2} \log \left( \frac{1+ h_{\Lambda} \Gamma_1}{1 + \Gamma_1 h_1 (1 + h_2 \Gamma_2)^{-1}}\right) \right]^+,\right. \nonumber \\
&\left. \phantom{----l-}
R_2 \leq \left[ \frac{1}{2} \log \left( \frac{1+ h_{\Lambda} \Gamma_2}{1 + \Gamma_2 h_2 (1 + h_1 \Gamma_1)^{-1}}\right) \right]^+, \right. \nonumber \\
&\left. \phantom{---}R_1 +R_2 \leq \left[ \frac{1}{2} \log \left( \frac{1+ h_{\Lambda} ( \Gamma_1 +\Gamma_2)}{1 + \Gamma_1 h_1 + \Gamma_2 h_2 }\right) \right]^+ \right\}.
\end{align}

We show in Theorem \ref{propstar5} that $(R_1^* (v),R_2^* (v))$ defined in Proposition \ref{Propaxi} is not only the unique allocation that satisfies the axioms of Definition \ref{defax2}, but also an allocation that is achievable and belongs to the~core.

\begin{thm} \label{propstar5}
Consider $(R_1^* (v),R_2^* (v))$ as defined in Proposition \ref{Propaxi}.
\begin{enumerate}[(i)]
\item $(R_1^* (v),R_2^* (v))$ belongs to $\mathcal{R}_a$ and is thus achievable. \label{propstar5i}
\item $(R_1^* (v),R_2^* (v))$  belongs to the core. \label{propstar5ii}
\end{enumerate}
\end{thm}
\begin{proof}
$(i)$ and $(ii)$ are proved in Appendices \ref{App_propstar5} and \ref{App_propstar6}, respectively.
\end{proof}

We illustrate Theorem \ref{propstar5} in Figure \ref{fig:ex23} with an example when $(\Gamma_1,\Gamma_2,h_1,h_2, \Lambda)= (1,0.4,0.6,0.8,0.1)$. Similar to the degraded case, we observe that some rates in the core may not be achievable. However, as proved in Theorem \ref{propstar5}, we observe in  Figure \ref{fig:ex23} that the unique fair allocation (that satisfies Definition \ref{defax2}) characterized in Proposition \ref{Propaxi} is achievable and belongs to the core.

 \begin{figure} 
\centering   
  \includegraphics[width=8.5cm]{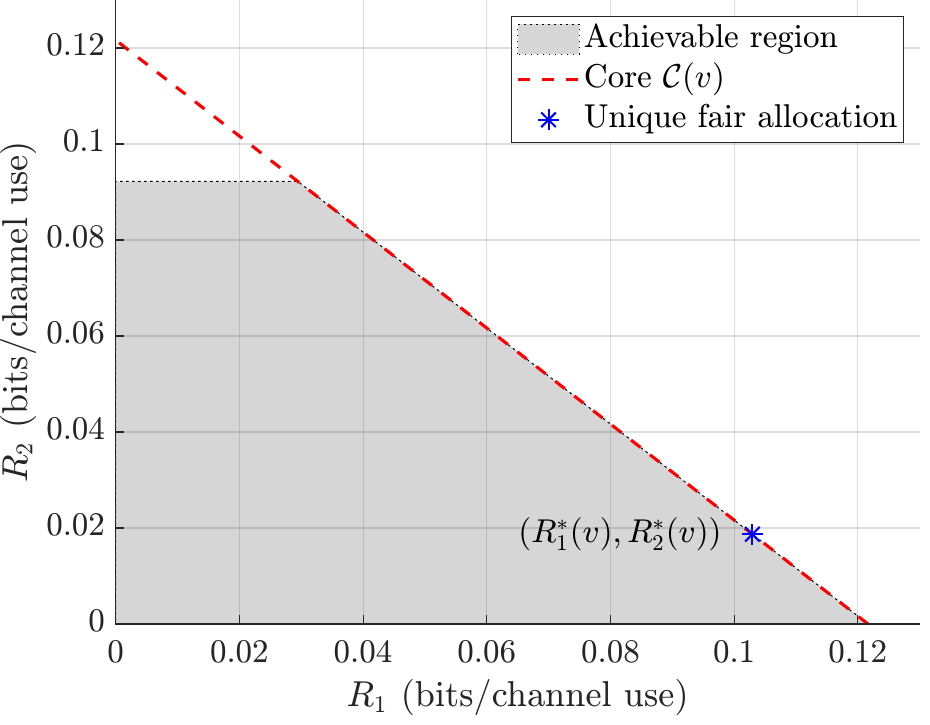}
  \caption{Representation of the known achievable region $\mathcal{R}_a$ from \cite{ISIT17}, the core $\mathcal{C}(v)$ defined in Definition \ref{defcore}, and the unique fair allocation $(R_1^*(v),R_2^*(v))$ defined in Proposition \ref{Propaxi} with the parameters $(\Gamma_1,\Gamma_2,h_1,h_2, \Lambda)= (1,0.4,0.6,0.8,0.1)$.} \label{fig:ex23}
\end{figure}

\section{Concluding Remarks} \label{sec:concl}
We studied the Gaussian multiple access wiretap channel with selfish transmitters.  Although a collective protocol can increase the individual secrecy rate of the transmitters, it can, at
the same time, lead to conflict of interests. We cast the problem as a coalitional game in which the value function is determined under information-theoretic
guarantees, i.e., the value associated with a coalition is computed with no restrictions on the
strategies that the transmitters outside the coalition can adopt. 

We showed for the degraded Gaussian multiple access wiretap channel that the grand coalition is in the best interest of all agents and stable,
in the sense that any coalition of transmitters has a disincentive to leave the grand coalition. We also
determined a fair secrecy rate allocation that is unique, achievable, and belongs to the core.

We also studied the non-degraded Gaussian multiple access wiretap channel with two selfish transmitters and determined that cooperation might not be in the best interest of all transmitters. We determined a necessary and sufficient condition for the core of the game to be non-empty. We also identified a range of channel parameters for which cooperation is always beneficial. In this case, we further proposed a solution concept for a fair allocation that satisfies a series of axioms that generalizes the one used for the degraded case, and determined a unique solution that is achievable and belongs to the core.

An open problem is the treatment of non-degraded Gaussian multiple access wiretap channels with an arbitrary number of transmitters. While the same coalitional game theory framework can still be applied to this case, the main difficulty lies in the characterization of the value function.

\appendices

\section{Proof of Theorem \ref{propcore}} \label{App_th4}

$(ii)$ follows from Theorem \ref{thregion}. 
We now show $(i)$, i.e., $\mathcal{C}^*(v) \subseteq \mathcal{C}(v)$. Let $\mathcal{S} \subseteq \mathcal{L}$ and assume $v(\mathcal{S})> 0$, i.e., $h_{\Lambda_{\mathcal{S}^c}}>h$. Define $\mathcal{S}(\Lambda) \triangleq \mathcal{S} \cap \mathcal{L} (\Lambda)$. Let $(R_l)_{l \in \mathcal{L}} \in \mathcal{C}^*(v)$, we have
\begin{subequations}
	\begin{align}
&R_{\mathcal{S}} \\
& = v(\mathcal{L}) - R_{\mathcal{S}^c}  \label{eq40a}\\
& \geq v(\mathcal{L}) - \frac{1}{2} \log \left[ \frac{1 + h_{\Lambda} \Gamma_{\mathcal{S}^c(\Lambda)}}{1 +h  \Gamma_{\mathcal{S}^c(\Lambda)}} \right] \label{eq40b}\\
& =  \frac{1}{2} \log \left[ \frac{1 + h_{\Lambda} \Gamma_{\mathcal{L}(\Lambda)} }{1 +  h_{\Lambda} \Gamma_{\mathcal{S}^c(\Lambda)}}  \frac{1 +h  \Gamma_{\mathcal{S}^c(\Lambda)}}{1 +h  \Gamma_{\mathcal{L}(\Lambda)}} \right] \\
& =  \frac{1}{2} \log \left[ \left(1 + \frac{ h_{\Lambda}\Gamma_{\mathcal{S}(\Lambda)} }{1 +  h_{\Lambda} \Gamma_{\mathcal{S}^c(\Lambda)} } \right)\left(1 + \frac{h \Gamma_{\mathcal{S}(\Lambda)} }{1 + h \Gamma_{\mathcal{S}^c(\Lambda)} } \right)^{-1} \right] \\
& \geq  \frac{1}{2} \log \left[ \left(1 + \frac{ \Gamma_{\mathcal{S}(\Lambda)} }{1 + \Lambda  +  \Gamma_{\mathcal{S}^c(\Lambda)} } \right)\left(1 + h \Gamma_{\mathcal{S}(\Lambda)}   \right)^{-1} \right] \\
& \geq  \frac{1}{2} \log \left[ \left(1 + \frac{ \Gamma_{\mathcal{S}(\Lambda)} }{1 +  \Lambda_{\mathcal{S}^c(\Lambda)} } \right)\left(1 + h \Gamma_{\mathcal{S}(\Lambda)}  \right)^{-1} \right] \label{eq40c} \\
& \geq  \frac{1}{2} \log \left[ \left(1 + \frac{ \Gamma_{\mathcal{S}(\Lambda_{\mathcal{S}^c})} }{1 +  \Lambda_{\mathcal{S}^c(\Lambda)} } \right)\left(1 + h \Gamma_{\mathcal{S}(\Lambda_{\mathcal{S}^c})}  \right)^{-1} \right] \label{eq40d} \\
& \geq v(\mathcal{S}) , \label{eq40e} 
\end{align}
\end{subequations}
where \eqref{eq40a} and \eqref{eq40b} hold by definition of $\mathcal{C}^*(v)$, \eqref{eq40c} holds by definition of $\Lambda_{\mathcal{S}^c}$, \eqref{eq40d} holds because $\mathcal{S}(\Lambda_{\mathcal{S}^c}) \subseteq \mathcal{S}(\Lambda)$ and when $h_{\Lambda_{\mathcal{S}^c(\Lambda)}}>h$, $x \mapsto \log\left(\frac{1+xh_{\Lambda_{\mathcal{S}^c(\Lambda)}}}{1+xh}\right)$ is increasing, \eqref{eq40e} holds because $\Lambda_{\mathcal{S}^c} \geq \Lambda_{\mathcal{S}^c(\Lambda)}$. Hence, $(R_l)_{l \in \mathcal{L}} \in \mathcal{C}(v)$. 

\section{Proof of Proposition \ref{prop1}} \label{App_prop1}
We consider $\Lambda =0$ to simplify notation in the proof, however, the case $\Lambda \neq 0$ is treated similarly. The proof of existence is similar to the one of \cite[Theorem 5.1]{la2004game}. Define for $x \in [0,1]$, for $l \in \mathcal{L}$
\begin{align}
\phi_{x,l}(v) \triangleq \frac{\frac{1}{2} \log \left[ 1+ x(l\Gamma_l + \Gamma_{l+1:L}) \right] - \sum_{i=l+1}^L \phi_{x,i}(v)}{l} . \label{eq:eff}
\end{align}
Some manipulations, similar to \cite[Eq.(8)]{la2004game}, gives that for any $x \in [0,1]$, for any $l \in \mathcal{L} \backslash \{L\}$
\begin{align}
\phi_{x,l}(v) - \phi_{x,l+1}(v) =  \frac{1}{2l} \log \left[ \frac{ 1 + x(l\Gamma_l  + \Gamma_{l+1:L})}{1 + x( l\Gamma_{l+1}  + \Gamma_{l+1:L} )} \right], \label{eq:symm}
\end{align}
and, as shown in \cite[Lemma 1]{la2004game}, that for any $x \in [0,1]$, for any $l, l' \in \mathcal{L}$ such that $\Gamma_l > \Gamma_l'$
\begin{align}
\phi_{x,l}(v^{l,l'}) = \phi_{x,l'}(v). \label{eq:envy}
\end{align}

Define now the following secrecy rate allocation for $l \in \mathcal{L}$
\begin{align}
R^*_l(v) \triangleq \phi_{h_{\Lambda},l}(v) -  \phi_{h,l}(v). \label{eq:alloc}
\end{align}
From \eqref{eq:alloc}, efficiency is seen by choosing $l=1$ in \eqref{eq:eff}, symmetry follows from \eqref{eq:symm}, and envy-freeness follows from~\eqref{eq:envy}. The proof of uniqueness is identical to the proof of \cite[Theorem~5.1]{la2004game}.

\section{Proof of Theorem \ref{th3}} \label{App_th3}

\subsection{Preliminaries}
 We assume the sequence $(\Gamma_{l})_{l\in \mathcal{L}}$ decreasing by relabeling the players if necessary. In the following we use the notation $\Gamma_{i:j} \triangleq \sum_{l=i}^j \Gamma_l$ for any $i,j \in \mathbb{N}$. We also define $L(\Lambda) \triangleq |\mathcal{L}(\Lambda)|$. We will use of the following lemma, which is proved in Appendix \ref{App_deriv}.
\begin{lem} \label{lemfunc}
Let $k \in \mathbb{N}^*$, $h_1 ,h_2 \in [0,1[$, such that $h_1 < h_2$, $a,b,c \in \mathbb{R}_+$  such that $b>a$. The following functions are non-increasing.
\begin{align}
  f^{(1)}_{k,h_1,h_2,a}: \phantom{l}\mathbb{R}_+ &\to  \mathbb{R}, \nonumber
  \\x &\mapsto \frac{1}{2} \log \left[ \frac{1 + h_1 (kx + a)}{1 + h_2(k x + a)} \right],&\\
  f^{(2)}_{k,h_1,h_2,a,b} : \phantom{l}\mathbb{R}_+  &\to  \mathbb{R}, \nonumber
  \\ x&\mapsto  f^{(1)}_{1,h_1,h_2,ka}(x) - f^{(1)}_{1,h_1,h_2,kb}(x),&\\
    f^{(3)}_{k,h_1,h_2,a}: \phantom{l}\mathbb{R}_+  &\to  \mathbb{R}, \nonumber
  \\ x&\mapsto  (k+1)  f^{(1)}_{k,h_1,h_2,a}(x) \nonumber
  \\ &\phantom{------} - k  f^{(1)}_{k+1,h_1,h_2,a}(x) ,&\\
        f^{(4)}_{k,h_1,h_2,a,c}: [0,c[ &\to  \mathbb{R}, \nonumber
  \\ x&\mapsto    f^{(1)}_{k+1,h_1,h_2,a}(x) \nonumber
  \\ &\phantom{---} - (k+1)  f^{(1)}_{1,h_1,h_2,a+ck}(x).&
\end{align}
Consequently, we also have that $\forall x \in \mathbb{R}^*_+, f^{(3)}_{k,h_1,h_2,a}(x) < 0$, since $f^{(3)}_{k,h_1,h_2,a}(0) \leq 0$. 
\end{lem}

\subsection{Left-hand side of \eqref{equation17}}
We first prove that for any $l\in\mathcal{L}(\Lambda)$, $R^*_l(v) >0$.
For $l \in \llbracket 1, L(\Lambda)-1 \rrbracket$,
\begin{subequations}
\begin{align}
&R^*_l(v) - R^*_{l+1}(v) \\
& = \frac{f^{(1)}_{l,h,h_{\Lambda}, \Gamma_{l+1:L(\Lambda)}}(\Gamma_{l+1})-f^{(1)}_{l,h,h_{\Lambda}, \Gamma_{l+1:L(\Lambda)}}(\Gamma_{l})}{l} \label{eq33a} \\
& \geq 0, \label{eq33b}
\end{align}
\end{subequations}
where \eqref{eq33a} holds by \eqref{eq:symm} and \eqref{eq:alloc}, and  \eqref{eq33b} holds because $\Gamma_{l+1} \leq \Gamma_l$ and $f^{(1)}_{l,h,h_{\Lambda}, \Gamma_{l+1:L(\Lambda)}}$ is non-increasing by Lemma~\ref{lemfunc}. 
Then, 
\begin{subequations}
\begin{align}
R^*_{L(\Lambda)}(v) 
&= -f^{(1)}_{L(\Lambda),h,h_{\Lambda},0}(\Gamma_{(\Lambda)}) \label{eq35a}  \\
& > 0,\label{eq35b} 
\end{align}
\end{subequations}
where \eqref{eq35a} holds by \eqref{eq:symm} and \eqref{eq:alloc}, \eqref{eq35b} holds because $h<h_{\Lambda}$.
Hence, by \eqref{eq33b} and \eqref{eq35b}, we have by induction that for any $l\in\mathcal{L}(\Lambda)$, $R^*_l(v) >0$. 

\subsection{Right-hand side of \eqref{equation17}}

Next, we want to prove that 
\begin{align} \label{eqstate}
\forall \mathcal{S} \subseteq \mathcal{L}(\Lambda), R^*_{\mathcal{S}}(v)   \leq \frac{1}{2} \log \left( \frac{1+ h_{\Lambda} \Gamma_{\mathcal{S}}}{1+ h \Gamma_{\mathcal{S}}}  \right).
\end{align}
 We prove \eqref{eqstate} by induction. Clearly, \eqref{eqstate} is true when $L(\Lambda)=1$. We assume that \eqref{eqstate} is true for $L(\Lambda)= K \in \mathbb{N}^*$, we will show that \eqref{eqstate} is true for $L(\Lambda) = K+1$. Let $v$ be the game with $L(\Lambda)=K+1$ transmitters. We let $v^{(-j)}$ denote the game $v$ by removing Transmitter $j \in \llbracket 1 , K+1 \rrbracket$. We first show the following lemma.

\begin{lem}
We have for any $l \in \llbracket 1 ,K+1\rrbracket$, for any $j \in \llbracket 1 ,K+1 \rrbracket \backslash \{ l\}$,
\begin{align}
R_l^*(v) < R_l^*(v^{(-j)}). \label{eq43}
\end{align}
	
\end{lem}

\begin{proof} 
For any $l \in \llbracket j+1 , K \rrbracket$
\begin{subequations}
\begin{align}
&R^*_l(v) - R^*_{l+1}(v) \\
& = \frac{f^{(1)}_{l,h,h_{\Lambda}, \Gamma_{l+1:K+1}}(\Gamma_{l+1})-f^{(1)}_{l,h,h_{\Lambda}, \Gamma_{l+1:K+1}}(\Gamma_{l})}{l} \label{eq36a} \\
& \leq \frac{f^{(1)}_{l-1,h,h_{\Lambda}, \Gamma_{l+1:K+1}}(\Gamma_{l+1})-f^{(1)}_{l-1,h,h_{\Lambda}, \Gamma_{l+1:K+1}}(\Gamma_{l})}{l-1} \label{eq36b} \\
& = R^*_l(v^{(-j)}) - R^*_{l+1}(v^{(-j)}), \label{eq36c}
\end{align}
\end{subequations}
where \eqref{eq36a} and \eqref{eq36c} hold by \eqref{eq:symm} and \eqref{eq:alloc}, \eqref{eq36b} holds because $\Gamma_{l+1} \leq \Gamma_l$ and  $f^{(3)}_{l,h,h_{\Lambda},\Gamma_{l+1:K+1}}$ is decreasing by Lemma \ref{lemfunc}. 

For any $l \in \llbracket 1 , j-2 \rrbracket$
\begin{subequations}
\begin{align}
&R^*_l(v) - R^*_{l+1}(v) \\
& = \frac{f^{(1)}_{l,h,h_{\Lambda}, \Gamma_{l+1:K+1}}(\Gamma_{l+1})-f^{(1)}_{l,h,h_{\Lambda}, \Gamma_{l+1:K+1}}(\Gamma_{l})}{l} \label{eq37a} \\
& \leq \frac{f^{(1)}_{l,h,h_{\Lambda}, \Gamma_{l+1:K+1} - \Gamma_j}(\Gamma_{l+1})-f^{(1)}_{l,h,h_{\Lambda}, \Gamma_{l+1:K+1}-\Gamma_j}(\Gamma_{l})}{l} \label{eq37b} \\
& = R^*_l(v^{(-j)}) - R^*_{l+1}(v^{(-j)}), \label{eq37c}
\end{align}
\end{subequations}
where \eqref{eq37a} and \eqref{eq37c} hold as \eqref{eq36a}, \eqref{eq37b} holds because $ [f^{(1)}_{l,h,h_{\Lambda}, \Gamma_{l+1:K+1} }(\Gamma_{l+1})- f^{(1)}_{l,h,h_{\Lambda}, \Gamma_{l+1:K+1}}(\Gamma_{l}))] - [f^{(1)}_{l,h,h_{\Lambda}, \Gamma_{l+1:K+1} - \Gamma_j}(\Gamma_{l+1})-f^{(1)}_{l,h,h_{\Lambda}, \Gamma_{l+1:K+1}-\Gamma_j}(\Gamma_{l})] = f^{(2)}_{l,h,h_{\Lambda}, \Gamma_{l+1},\Gamma_{l}}(\Gamma_{l+1:K+1}) - f^{(2)}_{l,h,h_{\Lambda}, \Gamma_{l+1},\Gamma_{l}}(\Gamma_{l+1:K+1}-\Gamma_j)$ and $f^{(2)}_{l,h,h_{\Lambda}, \Gamma_{l+1},\Gamma_{l}}$ is decreasing by Lemma~\ref{lemfunc}.

When $1<j<K+1$, we have for $l = j$
\begin{subequations}
\begin{align}
&R^*_{l-1}(v) - R^*_{l+1}(v) \\
& = R^*_{l-1}(v) - R^*_{l}(v) + R^*_{l}(v) - R^*_{l+1}(v) \\
& = \frac{f^{(1)}_{l-1,h,h_{\Lambda}, \Gamma_{l:K+1}}(\Gamma_{l})-f^{(1)}_{l-1,h,h_{\Lambda}, \Gamma_{l:K+1}}(\Gamma_{l-1})}{l-1} \nonumber \\
& \phantom{--}+\frac{f^{(1)}_{l,h,h_{\Lambda}, \Gamma_{l+1:K+1}}(\Gamma_{l+1})-f^{(1)}_{l,h,h_{\Lambda}, \Gamma_{l+1:K+1}}(\Gamma_{l})}{l} \label{eq31a} \\
& \leq \frac{1}{l-1}f^{(1)}_{l-1,h,h_{\Lambda}, \Gamma_{l+1} + \Gamma_{l+1:K+1}}(\Gamma_{l+1})\nonumber \\
&\phantom{--}-\frac{1}{l-1}f^{(1)}_{l-1,h,h_{\Lambda}, \Gamma_{l+1} + \Gamma_{l+1:K+1}}(\Gamma_{l-1})  \label{eq31b} \\
& \leq \frac{f^{(1)}_{l-1,h,h_{\Lambda}, \Gamma_{l+1:K+1}}(\Gamma_{l+1})-f^{(1)}_{l-1,h,h_{\Lambda},  \Gamma_{l+1:K+1}}(\Gamma_{l-1})}{l-1}  \label{eq31c} \\
& = R^*_{l-1}(v^{(-j)}) - R^*_{l+1}(v^{(-j)}), \label{eq31d}
\end{align}
\end{subequations}
where \eqref{eq31a} and \eqref{eq31d} hold as \eqref{eq36a}, \eqref{eq31b} holds because 
\begin{align}
&\frac{f^{(1)}_{l-1,h,h_{\Lambda}, \Gamma_{l:K+1}}(\Gamma_{l})-f^{(1)}_{l-1,h,h_{\Lambda}, \Gamma_{l:K+1}}(\Gamma_{l-1})}{l-1} \nonumber \\
&+\frac{f^{(1)}_{l,h,h_{\Lambda}, \Gamma_{l+1:K+1}}(\Gamma_{l+1})-f^{(1)}_{l,h,h_{\Lambda}, \Gamma_{l+1:K+1}}(\Gamma_{l})}{l}\nonumber \\
& - \frac{f^{(1)}_{l-1,h,h_{\Lambda}, \Gamma_{l+1} + \Gamma_{l+1:K+1}}(\Gamma_{l+1})}{l-1} \nonumber \\
&
+\frac{f^{(1)}_{l-1,h,h_{\Lambda}, \Gamma_{l+1}+ \Gamma_{l+1:K+1}}(\Gamma_{l-1})}{l-1} \nonumber \\
& = \frac{1}{l(l-1)} (f^{(4)}_{l-1,h,h_{\Lambda},\Gamma_{l+1:K+1},\Gamma_{l-1}}(\Gamma_l)\nonumber \\
&\phantom{--}-f^{(4)}_{l-1,h,h_{\Lambda},\Gamma_{l+1:K+1},\Gamma_{l-1}}(\Gamma_{l+1})),\end{align} and because $f^{(4)}_{l-1,h,h_{\Lambda},\Gamma_{l+1:K+1},\Gamma_{l-1}}$ is decreasing by Lemma~\ref{lemfunc}, \eqref{eq31c} holds as \eqref{eq37b}.

 When $j \neq K+1$, we have
\begin{subequations}
\begin{align}
R^*_{K+1}(v)  
& = - \frac{f^{(1)}_{K+1,h,h_{\Lambda},0}(\Gamma_{K+1})}{K+1} \\
& < - \frac{f^{(1)}_{K,h,h_{\Lambda},0}(\Gamma_{K+1})}{K} \label{eq41a}\\
& = R^*_{K+1}(v^{(-j)}),\label{eq41b}
\end{align}
\end{subequations}
where \eqref{eq41a} holds because $\forall x \in \mathbb{R}^*_+, f^{(3)}_{K,h,h_{\Lambda},0}(x) < 0$ by Lemma \ref{lemfunc}.

When $j=K+1$, we have
\begin{subequations}
\begin{align}
&R^*_{K}(v)  \nonumber \\
& = \frac{1}{K} \left[ - f^{(1)}_{K,h,h_{\Lambda},\Gamma_{K+1}} (\Gamma_K) + \frac{f^{(1)}_{K+1,h,h_{\Lambda},0} (\Gamma_{K+1})}{K+1} \right] \label{eq38aa}\\
& \leq \frac{1}{K} \left[ -f^{(1)}_{K,h,h_{\Lambda},0}(\Gamma_{K})   - f^{(1)}_{K+1,h,h_{\Lambda},0}(\Gamma_{K+1}) \right. \nonumber \\
& \phantom{----}\left.+ f^{(1)}_{K,h,h_{\Lambda},0} (\Gamma_{K+1}) + \frac{f^{(1)}_{K+1,h,h_{\Lambda},0} (\Gamma_{K+1})}{K+1} \right] \label{eq38a} \\
& = \frac{1}{K} \left[ -f^{(1)}_{K,h,h_{\Lambda},0}(\Gamma_{K})   + f^{(1)}_{K,h,h_{\Lambda},0} (\Gamma_{K+1})\right. \nonumber \\
& \phantom{----}\left. -K \frac{f^{(1)}_{K+1,h,h_{\Lambda},0} (\Gamma_{K+1})}{K+1} \right]  \\
& < - \frac{f^{(1)}_{K,h,h_{\Lambda},0}(\Gamma_{K})}{K} \label{eq38b}\\
& = R^*_{K}(v^{(-j)}),\label{eq38c}
\end{align}
\end{subequations}
\eqref{eq38aa} and \eqref{eq38c} hold by \eqref{eq:alloc} and \eqref{eq:eff}, \eqref{eq38a} holds because $f^{(1)}_{K,h,h_{\Lambda},0}(\Gamma_{K}) - f^{(1)}_{K,h,h_{\Lambda},\Gamma_{K+1}} (\Gamma_{K}) + f^{(1)}_{K+1,h,h_{\Lambda},0}(\Gamma_{K+1}) - f^{(1)}_{K,h,h_{\Lambda},0} (\Gamma_{K+1}) = f^{(2)}_{1,h,h_{\Lambda},0,\Gamma_{K+1}} (K \Gamma_K) -f^{(2)}_{1,h,h_{\Lambda},0,\Gamma_{K+1}} (K \Gamma_{K+1}) \leq 0 $ and because  $f^{(2)}_{1,h,h_{\Lambda},0,\Gamma_{K+1}}$ is decreasing by Lemma \ref{lemfunc}, \eqref{eq38b} holds because $f^{(3)}_{K,h,h_{\Lambda},0}$ is strictly negative by Lemma \ref{lemfunc}.

Hence, by \eqref{eq36c}, \eqref{eq37c}, \eqref{eq31d}, \eqref{eq41b} when $j\neq K+1$, and by \eqref{eq36c}, \eqref{eq37c},  \eqref{eq38c} when $j =K+1$, we have \eqref{eq43} for any $l \in \llbracket 1 ,K+1\rrbracket$, for any $j \in \llbracket 1 ,K+1 \rrbracket \backslash \{ l\}$.
\end{proof}
Finally, for any $\mathcal{S} \subsetneq \llbracket 1 , K+1 \rrbracket$, there exists $j_{\mathcal{S}} \in \llbracket 1,K+1 \rrbracket \backslash \mathcal{S}$ such that
\begin{subequations}
\begin{align}
R^*_{\mathcal{S}}(v) 
&< R^*_{\mathcal{S}}(v^{(-j_{\mathcal{S}})}) \label{eq44a} \\ 
& < \frac{1}{2} \log \left( \frac{1+ h_{\Lambda} \Gamma_{\mathcal{S}}}{1+ h \Gamma_{\mathcal{S}}}  \right), \label{eq44b}
\end{align}
\end{subequations}
where \eqref{eq44a} holds by \eqref{eq43}, and \eqref{eq44b} holds by induction hypothesis.

\section{Proof of Proposition \ref{prop2}} \label{App_prop2}

\subsection{Monotonicity}

We will use of the following lemma, which is proved in Appendix \ref{App_deriv2}, to prove by induction, that for any $l \in \llbracket 1, L(\Lambda)-1 \rrbracket$, $\frac{R_{l}^*(v^{(\omega)}) }{R_{l+1}^*(v^{(\omega)}) }$ is an increasing function of $\omega$. 
\begin{lem} \label{lemfunc2}
Let $a,b,c,d \in \mathbb{R}^*_+$, $h_1 ,h_2 \in [0,1[$, such that $h_1 < h_2$. The following functions are non-decreasing.
\begin{align} 
f_{h_1,h_2,a,b,c,d}^{(5)} : \mathbb{R}_+ &\to \mathbb{R},\nonumber\\ \label{eqdef5}
\omega &\mapsto \frac{\log \left[ \frac{\omega + h_1(a+b+c)}{\omega + h_2(a+b+c)}\frac{\omega + h_2(a+b+c+d)}{\omega + h_1(a+b+c+d)} \right]}{\log \left[ \frac{\omega + h_1a}{\omega + h_2 a}\frac{\omega + h_2(a+b)}{\omega + h_1(a+b)} \right]},\\
f_{h_1,h_2,a,b}^{(6)} : \mathbb{R}_+ &\to \mathbb{R}, \nonumber \\
\omega &\mapsto \frac{\log \left[ \frac{\omega + h_2(a+b)}{\omega + h_1(a+b)} \right]}{\log \left[ \frac{\omega + h_2 a}{\omega + h_1 a} \right]}.
\end{align}
\end{lem}

	From \eqref{eq33a} and Lemma \ref{lemfunc2}, we have that for any $l \in \llbracket 1, L(\Lambda)-1 \rrbracket$,  $l' \in \llbracket 1, L(\Lambda)-2 \rrbracket$,  such that $l' >l$, $\Gamma_{l'} \neq \Gamma_{l'+1}$, and $\Gamma_{l} \neq \Gamma_{l+1}$, we have
	\begin{align} \label{eqinc}
	\frac{R_l^*(v^{(\omega)}) - R_{l+1}^*(v^{(\omega)})}{R_{l'}^*(v^{(\omega)}) - R_{l'+1}^*(v^{(\omega)})} = \frac{l'}{l} f_{h,h_{\Lambda},a,b,c,d}^{(5)}(\omega),
	\end{align}
	where $a \triangleq l'\Gamma_{l'+1} + \Gamma_{l'+1:L(\Lambda)}$, $b \triangleq l'(\Gamma_{l'} - \Gamma_{l'+1})$, $c \triangleq l(\Gamma_{l+1} - \Gamma_{l'}) - (l'-l)\Gamma_{l'} + \Gamma_{l+1:l'}$, $d \triangleq l(\Gamma_{l} - \Gamma_{l+1}) $. Observe that $a,b,c,d \in \mathbb{R}^*_+$.
	Next, from \eqref{eq33a} and \eqref{eqoptalloc} we have 
	\begin{align} \label{eqinc2}
		&\frac{R_{L(\Lambda)-1}^*(v^{(\omega)}) }{R_{L(\Lambda)}^*(v^{(\omega)}) } \nonumber\\& = \frac{1}{L(\Lambda)-1} \left[ -1 \right. \nonumber \\
  & \left. \phantom{--}+L(\Lambda) f^{(6)}_{h,h_{\Lambda},L(\Lambda)\Gamma_{L(\Lambda)},(L(\Lambda)-1)(\Gamma_{L(\Lambda)-1}-\Gamma_{L(\Lambda)})}(\omega) \right].
			\end{align}
	Hence, by induction, using \eqref{eqinc}, \eqref{eqinc2}, and Lemma \ref{lemfunc2} for any $l \in \llbracket 1, L(\Lambda)-1 \rrbracket$, $\frac{R_{l}^*(v^{(\omega)}) }{R_{l+1}^*(v^{(\omega)}) }$ is an increasing function of $\omega$.
	\subsection{Image}

	We now determine the image of  $\omega \mapsto \frac{R_{l}^*(v^{(\omega)}) }{R_{l+1}^*(v^{(\omega)}) }$. 
	 For any $l \in \llbracket 1, L(\Lambda)-1 \rrbracket$,  any $l' \in \llbracket 1, L(\Lambda)-1 \rrbracket$ such that $\Gamma_{l'+1} \neq \Gamma_{l'}$, we have
	 \begin{subequations}
	\begin{align}
		&\frac{R^*_{l}(v^{(\omega)}) - R^*_{l+1}(v^{(\omega)})}{R^*_{l'}(v^{(\omega)}) - R^*_{l'+1}(v^{(\omega)})} \nonumber \\
		& = \frac{l'}{l} \frac{f^{(1)}_{l,h,h_{\Lambda}, \omega^{-1}\Gamma_{l+1:L(\Lambda)}}\!(\frac{\Gamma_{l+1}}{\omega})\!-\!f^{(1)}_{l,h,h_{\Lambda}, \omega^{-1}\Gamma_{l+1:L(\Lambda)}}\!(\frac{\Gamma_{l}}{\omega})}{f^{(1)}_{l',h,h_{\Lambda}, \omega^{-1}\Gamma_{l'+1:L(\Lambda)}}\!(\frac{\Gamma_{l'+1}}{\omega})\!-\!f^{(1)}_{l',h,h_{\Lambda}, \omega^{-1}\Gamma_{l'+1:L\!(\Lambda)}}(\frac{\Gamma_{l'}}{\omega})}  \label{eq47a} \\
		& \xrightarrow{\omega \to +\infty} \frac{\Gamma_{l+1} - \Gamma_{l}}{\Gamma_{l'+1} - \Gamma_{l'}}, \label{eq47b} 
	\end{align}
	\end{subequations}
	where \eqref{eq47a} holds as \eqref{eq36a}, and \eqref{eq47b} is obtained with Taylor series for $\log$.
	Similar to \eqref{eq47b}, we have
	\begin{align}
				\frac{R^*_{L(\Lambda)}(v^{(\omega)}) }{R^*_{L(\Lambda)-1}(v^{(\omega)}) } 
	\xrightarrow{\omega \to +\infty} \frac{\Gamma_{L(\Lambda)}}{\Gamma_{L(\Lambda)-1}}, \label{eq48} 
	\end{align}
	and by \eqref{eq47b} and \eqref{eq48}, we have by induction for any $l \in \llbracket 1, K-2 \rrbracket$\begin{align}
				\frac{R^*_{l+1}(v^{(\omega)}) }{R^*_{l}(v^{(\omega)}) } 
	\xrightarrow{\omega \to +\infty} \frac{\Gamma_{l+1}}{\Gamma_{l}}. 
	\end{align}
	Finally, we see from \eqref{eqoptalloc} that $R_{L(\Lambda)}^*(v^{(\omega)}) \xrightarrow{\omega \to 0} \frac{1}{2L(\Lambda)} \log \frac{h_{\Lambda}}{h}$, then by induction and using~\eqref{eqoptalloc} we get for any $l \in \mathcal{L}(\Lambda)$
	\begin{align}
	R_l^*(v^{(\omega)}) \xrightarrow{\omega \to 0} \frac{1}{2L(\Lambda)} \log \frac{h_{\Lambda}}{h},
	\end{align}
such that for any $l,l' \in \mathcal{L}(\Lambda)$	
\begin{align}
				\frac{R^*_{l+1}(v^{(\omega)}) }{R^*_{l}(v^{(\omega)}) } 
	\xrightarrow{\omega \to 0} 1. 
	\end{align}

\section{Proof of Lemma \ref{lemstar}} \label{App_lemstar}
If we define
\begin{align}
\Gamma_1^* &\triangleq \Gamma_2 \frac{ h_{\Lambda} - h_2}{ h_{\Lambda} - h_1 + 2 h_{\Lambda} \Gamma_2 (h_2 - h_1)},\label{47a} \\
\Gamma_2^* 
&\triangleq \Gamma_1 \frac{ h_{\Lambda} - h_1}{ h_{\Lambda} - h_2 + 2 h_{\Lambda} \Gamma_1 (h_1 - h_2)}, \label{47b}
\end{align}
then 
\begin{align}
v_{\Gamma_1 \leftarrow \Gamma_{1}^*, h_1 \leftarrow h_1}( \mathcal{L} ) 
&= \frac{1 + h_{\Lambda} (\Gamma_1^* +  \Gamma_2)}{1+ h_1\Gamma_1^* + h_2 \Gamma_2} \nonumber\\
& =  \frac{1 + 2h_{\Lambda}  \Gamma_2}{1+ 2 h_2 \Gamma_2} \nonumber\\
&=v_{\Gamma_1 \leftarrow \Gamma_2, h_1 \leftarrow h_2}( \mathcal{L} ),  \label{eqfndg}\\
v_{\Gamma_2 \leftarrow \Gamma_{2}^*, h_2 \leftarrow h_2}( \mathcal{L} ) 
& = \frac{1 + h_{\Lambda} (\Gamma_1 +  \Gamma_2^*)}{1+ h_1\Gamma_1 + h_2 \Gamma_2^*} \nonumber\\
& =  \frac{1 + 2h_{\Lambda}  \Gamma_1}{1+ 2 h_1 \Gamma_1} \nonumber\\
&=v_{\Gamma_2 \leftarrow \Gamma_1, h_2 \leftarrow h_1}( \mathcal{L} ).\end{align}

We now show by contradiction that $\Gamma_{1}^* \in ]0, \Gamma_1]$ or  $\Gamma_{2}^* \in ]0, \Gamma_2]$. Assume  that $\Gamma_{1}^* \notin ]0, \Gamma_1]$ and  $\Gamma_{2}^* \notin ]0, \Gamma_2]$. We consider three cases.

\textbf{Case 1}: $\Gamma_1^* > \Gamma_1$ and $\Gamma_2^* > \Gamma_2$. We obtain a contradiction as follows
\begin{subequations}
\begin{align}
&\Gamma_2 ( h_{\Lambda} - h_2) \nonumber \\
& >  \Gamma_1( h_{\Lambda} - h_1) + 2 h_{\Lambda} \Gamma_1 \Gamma_2 (h_2 - h_1) \label{51a} \\
& > \Gamma_2( h_{\Lambda} - h_2) + 2 h_{\Lambda} \Gamma_1 \Gamma_2 (h_1 - h_2) + 2 h_{\Lambda} \Gamma_1 \Gamma_2 (h_2 - h_1) \label{51b} \\
& = \Gamma_2 ( h_{\Lambda} - h_2),
\end{align}
\end{subequations}
where \eqref{51a} holds by \eqref{47a} and because $\Gamma_1^* > \Gamma_1$, \eqref{51b} holds by~\eqref{47b} and because $\Gamma_2^* > \Gamma_2$.

\textbf{Case 2}: $\Gamma_1^* > \Gamma_1$ and $\Gamma_2^* \leq 0$. (The case $\Gamma_2^* > \Gamma_2$ and $\Gamma_1^* \leq 0$ is obtained similarly by exchanging the role of Transmitter $1$ and Transmitter $2$.) We have $0 \geq  h_{\Lambda} - h_2 + 2 h_{\Lambda} \Gamma_1 (h_1 - h_2)$ by \eqref{47b} and because $\Gamma_2^* \leq 0$ and $h_{\Lambda} > h_1$. Hence,
\begin{subequations}
\begin{align}
&0 \nonumber\\
& \geq  \Gamma_2(h_{\Lambda} - h_2) + 2 h_{\Lambda} \Gamma_1 \Gamma_2 (h_1 - h_2)\label{55} \\
& > \Gamma_1( h_{\Lambda} - h_1) + 2 h_{\Lambda} \Gamma_1 \Gamma_2 (h_2 - h_1) + 2 h_{\Lambda} \Gamma_1 \Gamma_2 (h_1 - h_2)  \label{55a} \\
& =  \Gamma_1( h_{\Lambda} - h_1). 
\end{align}
\end{subequations}
where \eqref{55a} holds as \eqref{51a}. We thus obtain a contradiction since $h_{\Lambda} > h_1$.

\textbf{Case 3}: $\Gamma_1^* \leq 0$ and $\Gamma_2^* \leq 0$. Similar to \eqref{55}, we have
\begin{align}
0
& \geq  \Gamma_2(h_{\Lambda} - h_2) + 2 h_{\Lambda} \Gamma_1 \Gamma_2 (h_1 - h_2)  ,\\
0
& \geq  \Gamma_1(h_{\Lambda} - h_1) + 2 h_{\Lambda} \Gamma_1 \Gamma_2 (h_2 - h_1),  
\end{align}
which combined together gives
\begin{align}
0
& \geq  \Gamma_2(h_{\Lambda} - h_2) +  \Gamma_1(h_{\Lambda} - h_1) ,  
\end{align}
which in turn contradicts that $h_{\Lambda}> h_1,h_2$ and $\Gamma_1,\Gamma_2 >0$.

Finally, assume that $(\Gamma_{1}^*,\Gamma_{2}^*) \neq (\Gamma_{1},\Gamma_{2})$. We obtain that $\Gamma_{1}^* \in ]0, \Gamma_1]$ or (mutually exclusive or) $\Gamma_{2}^* \in ]0, \Gamma_2]$ by showing that we cannot have  $\Gamma_{1}^* \in ]0, \Gamma_1]$ and  $\Gamma_{2}^* \in ]0, \Gamma_2]$. Indeed, assume that 
$\Gamma_{1}^* \in ]0, \Gamma_1]$ and  $\Gamma_{2}^* \in ]0, \Gamma_2]$, then 
\begin{subequations}
\begin{align}
&\Gamma_2 ( h_{\Lambda} - h_2) \nonumber \\
& \leq  \Gamma_1( h_{\Lambda} - h_1) + 2 h_{\Lambda} \Gamma_1 \Gamma_2 (h_2 - h_1) \label{60a}  \\
& \leq  \Gamma_2( h_{\Lambda} - h_2) + 2 h_{\Lambda} \Gamma_1 \Gamma_2 (h_1 - h_2) + 2 h_{\Lambda} \Gamma_1 \Gamma_2 (h_2 - h_1)   \label{60b} \\
& = \Gamma_2( h_{\Lambda} - h_2),
\end{align}
\end{subequations}
where $\eqref{60a}$ holds by \eqref{47a} and because $\Gamma_1^* \in ]0, \Gamma_1]$ and $h_{\Lambda} > h_2$, $\eqref{60b}$ holds by \eqref{47b} and because $\Gamma_2^* \in ]0, \Gamma_2]$ and $h_{\Lambda} > h_1$. Since $(\Gamma_{1}^*,\Gamma_{2}^*) \neq (\Gamma_{1},\Gamma_{2})$, either \eqref{60a} or \eqref{60b} is a strict inequality and we obtain a contradiction.

\section{Proof of Theorem \ref{propstar5}.\ref{propstar5i}} \label {App_propstar5}

Clearly $R_2^* \geq 0$ and we have 
\begin{subequations}
\begin{align}
&\frac{1}{2} \log \left( \frac{1+ h_{\Lambda} \Gamma_2}{1 + \Gamma_2 h_2 (1 + h_1 \Gamma_1)^{-1}}\right)\\
& \geq \frac{1}{2} \log \left( \frac{1+ h_{\Lambda} \Gamma_2}{1 + \Gamma_2 h_2 }\right) \\
& \geq \frac{1}{2} \left[ \frac{1}{2} \log \left( \frac{1 + 2  h_{\Lambda}\Gamma_2}{1 + 2 h_2 \Gamma_2} \right) \right] \label{eq81b} \\
& = R_2^*(v),
\end{align}
\end{subequations}
where we have used in \eqref{eq81b} concavity of $x \mapsto \log \left( \frac{1 +   h_{\Lambda}x}{1 + h_2 x} \right)$.

Next we show that $R_1^* \geq 0$.
\begin{subequations}
\begin{align}
R_1^* (v) 
& =  \frac{1}{2} \log \left( \frac{1 +   h_{\Lambda}(\Gamma_1+ \Gamma_2)}{1  + h_1 \Gamma_1 + h_2 \Gamma_2} \right) \nonumber \\
& \phantom{--}-   \frac{1}{2} \left[ \frac{1}{2} \log \left( \frac{1 + 2  h_{\Lambda}\Gamma_2}{1 + 2 h_2 \Gamma_2} \right) \right]\\
& =  \frac{1}{2} \log \left( \frac{1 +   h_{\Lambda}(\Gamma_1+ \Gamma_2)}{1 + h_1 \Gamma_1 + h_2 \Gamma_2} \right)\nonumber \\
& \phantom{--} -   \frac{1}{2} \left[ \frac{1}{2} \log \left( \frac{1 + h_{\Lambda} (\Gamma_1^* +  \Gamma_2)}{1+ h_1\Gamma_1^* + h_2 \Gamma_2}
 \right) \right] \label{eq82b} \\
 & =  \frac{1}{2} \log \left( \frac{1 +   h_{\Lambda}(\Gamma_1+ \Gamma_2)}{1 + h_1 \Gamma_1 + h_2 \Gamma_2} \frac{1+ h_1\Gamma_1^* + h_2 \Gamma_2}{1 + h_{\Lambda} (\Gamma_1^* +  \Gamma_2)} \right) \nonumber \\
& \phantom{--}+    \frac{1}{4} \log \left( \frac{1 + h_{\Lambda} (\Gamma_1^* +  \Gamma_2)}{1+ h_1\Gamma_1^* + h_2 \Gamma_2}
 \right) ,
\end{align}
\end{subequations}
where \eqref{eq82b} holds by \eqref{eqfndg}.
Next, we show that $\frac{1}{2} \log \left( \frac{1 +   h_{\Lambda}(\Gamma_1+ \Gamma_2)}{1 + h_1 \Gamma_1 + h_2 \Gamma_2} \frac{1+ h_1\Gamma_1^* + h_2 \Gamma_2}{1 + h_{\Lambda} (\Gamma_1^* +  \Gamma_2)} \right) \geq 0$. We have
\begin{subequations}
\begin{align}
&\frac{1 +   h_{\Lambda}(\Gamma_1+ \Gamma_2)}{1 + h_1 \Gamma_1 + h_2 \Gamma_2} \frac{1+ h_1\Gamma_1^* + h_2 \Gamma_2}{1 + h_{\Lambda} (\Gamma_1^* +  \Gamma_2)} \geq 1 \\
&\Leftrightarrow (1 +   h_{\Lambda}(\Gamma_1+ \Gamma_2)) (1+ h_1\Gamma_1^* + h_2 \Gamma_2) \nonumber \\
& \phantom{--}- (1 + h_1 \Gamma_1 + h_2 \Gamma_2)(1 + h_{\Lambda} (\Gamma_1^* +  \Gamma_2)) \geq 0\\
&\Leftrightarrow \left(\Gamma _1-\Gamma^*_1\right) \left( h_{\Lambda
   }-h_1 + \Gamma _2 
   h_{\Lambda } ( h_2 -h_1)\right) \geq 0.
\end{align}
\end{subequations}
We consider two cases.

\textbf{Case 1}: Assume $(h_2-h_1)\geq 0$. Then $\left(\Gamma _1-\Gamma^*_1\right) \left( h_{\Lambda
   }-h_1 + \Gamma _2 
   h_{\Lambda } ( h_2 -h_1)\right) \geq 0$ because $h_{\Lambda} > h_1$ and $\Gamma_1 \geq \Gamma_1^*$.
   
\textbf{Case 2}: Assume $(h_2-h_1)< 0$. Then 
\begin{subequations}
\begin{align}
&\left(\Gamma _1-\Gamma^*_1\right) \left( h_{\Lambda
   }-h_1 + \Gamma _2 
   h_{\Lambda } ( h_2 -h_1)\right) \nonumber \\
&   \geq \left(\Gamma _1-\Gamma^*_1\right) \left( h_{\Lambda
   }-h_1 + 2 \Gamma _2 
   h_{\Lambda } ( h_2 -h_1)\right)\\
   & \geq 0,
\end{align}
\end{subequations}
by using \eqref{47a} with the fact that $h_{\Lambda} > h_2$ and $ \Gamma_1^* \geq 0$ .

We thus have that $\frac{1}{2} \log \left( \frac{1 +   h_{\Lambda}(\Gamma_1+ \Gamma_2)}{1 + h_1 \Gamma_1 + h_2 \Gamma_2} \frac{1+ h_1\Gamma_1^* + h_2 \Gamma_2}{1 + h_{\Lambda} (\Gamma_1^* +  \Gamma_2)} \right) \geq 0$. We also have $\frac{1}{4} \log \left( \frac{1 + h_{\Lambda} (\Gamma_1^* +  \Gamma_2)}{1+ h_1\Gamma_1^* + h_2 \Gamma_2}
 \right) \geq 0$ by \eqref{eqfndg}.
We deduce that $R_1^* \geq 0$.

Next, we have
\begin{subequations}
\begin{align}
&R_1^* (v)  \nonumber \\
& =  \frac{1}{2} \log \left( \frac{1 +   h_{\Lambda}(\Gamma_1+ \Gamma_2)}{1  + h_1 \Gamma_1 + h_2 \Gamma_2} \right) -   \frac{1}{2} \left[ \frac{1}{2} \log \left( \frac{1 + 2  h_{\Lambda}\Gamma_2}{1 + 2 h_2 \Gamma_2} \right) \right]\\
& =  \frac{1}{4} \log \left[ \frac{ (1 +   h_{\Lambda}(\Gamma_1+ \Gamma_2))^2}{(1  + h_2 \Gamma_2)^2(1  + h_{\Lambda} \Gamma_1)^2} \frac{1 + 2 h_2 \Gamma_2}{1 + 2  h_{\Lambda}\Gamma_2} \right] \nonumber \\
& \phantom{--}+ \frac{1}{2} \log \left( \frac{1+ h_{\Lambda} \Gamma_1}{1 + \Gamma_1 h_1 (1 + h_2 \Gamma_2)^{-1}}\right). \label{eq85c}
\end{align}
\end{subequations}
Then,  we have
\begin{subequations}
\begin{align}
&\frac{ (1 +   h_{\Lambda}(\Gamma_1+ \Gamma_2))^2}{(1  + h_2 \Gamma_2)^2(1  + h_{\Lambda} \Gamma_1)^2} \frac{1 + 2 h_2 \Gamma_2}{1 + 2  h_{\Lambda}\Gamma_2} \leq 1  \nonumber\\
&\Leftrightarrow (1 +   h_{\Lambda}(\Gamma_1+ \Gamma_2))^2 (1 + 2 h_2 \Gamma_2) \nonumber\\
&\phantom{--}- (1  + h_2 \Gamma_2)^2(1  + h_{\Lambda} \Gamma_1)^2 (1 + 2  h_{\Lambda}\Gamma_2)
\leq 0\\
&\Leftrightarrow -\Gamma _2 \left(2 \Gamma _1^2 h_{\Lambda }^3+2 \Gamma _1^2 \Gamma _2^2 h_2^2 h_{\Lambda }^3+4 \Gamma _1^2 \Gamma _2 h_2 h_{\Lambda }^3\right. \nonumber\\
   & \phantom{--}\left.+4 \Gamma _1 \Gamma _2^2 h_2^2
   h_{\Lambda }^2+2 \Gamma _1 h_{\Lambda }^2 + \Gamma _1^2 \Gamma _2 h_2^2 h_{\Lambda }^2+4 \Gamma _1 \Gamma _2 h_2 h_{\Lambda }^2\right. \nonumber\\
   & \phantom{--}\left. +2 \Gamma _1 \Gamma _2 h_2^2
   h_{\Lambda }+\Gamma _2 \left(h_2-h_{\Lambda }\right) \left(2 \Gamma _2 h_2 h_{\Lambda }+h_{\Lambda }+h_2\right)\right)\nonumber\\
   & \phantom{--} \leq 0.
\end{align}
\end{subequations}

By definition of $\Gamma_1^*$ in \eqref{47a} and because $\Gamma_1^* \in ]0, \Gamma_1]$, we have
\begin{align}
  \Gamma_2(h_2 - h_{\Lambda} ) \geq -  \Gamma_1 ( h_{\Lambda} - h_1 + 2 h_{\Lambda} \Gamma_2 (h_2 - h_1)). \label{eq87ff}
\end{align}

Hence, we have
\begin{subequations}
\begin{align}
&-\Gamma _2 \left(2 \Gamma _1^2 h_{\Lambda }^3+2 \Gamma _1^2 \Gamma _2^2 h_2^2 h_{\Lambda }^3+4 \Gamma _1^2 \Gamma _2 h_2 h_{\Lambda }^3\right. \nonumber\\
   & \phantom{}\left.+4 \Gamma _1 \Gamma _2^2 h_2^2
   h_{\Lambda }^2+2 \Gamma _1 h_{\Lambda }^2 + \Gamma _1^2 \Gamma _2 h_2^2 h_{\Lambda }^2+4 \Gamma _1 \Gamma _2 h_2 h_{\Lambda }^2\right. \nonumber\\
   & \phantom{}\left. +2 \Gamma _1 \Gamma _2 h_2^2
   h_{\Lambda }+\Gamma _2 \left(h_2-h_{\Lambda }\right) \left(2 \Gamma _2 h_2 h_{\Lambda }+h_{\Lambda }+h_2\right)\right) \nonumber\\
 &  \leq -\Gamma _2 \left(2 \Gamma _1^2 h_{\Lambda }^3+2 \Gamma _1^2 \Gamma
   _2^2 h_2^2 h_{\Lambda }^3+4 \Gamma _1^2 \Gamma _2 h_2 h_{\Lambda
   }^3 \right. \nonumber \\
   &\left. \phantom{-}+4 \Gamma _1 \Gamma _2^2 h_1 h_2 h_{\Lambda }^2+\Gamma _1   h_{\Lambda } ( h_{\Lambda }  -  h_2 ) +\Gamma _1^2 \Gamma _2 h_2^2 h_{\Lambda }^2\right. \nonumber \\
   &\left. \phantom{-}+2 \Gamma
   _1 \Gamma _2 h_1 h_{\Lambda }^2+\Gamma _1 h_1 h_{\Lambda }+4 \Gamma _1 \Gamma _2 h_1 h_2 h_{\Lambda
   }+\Gamma _1 h_1 h_2\right) \nonumber\\
   & \leq 0,
\end{align}
since $h_{\Lambda }  >  h_2$.
\end{subequations}
We conclude with \eqref{eq85c} that 
\begin{align}
R_1^* (v)   \leq   \frac{1}{2} \log \left( \frac{1+ h_{\Lambda} \Gamma_1}{1 + \Gamma_1 h_1 (1 + h_2 \Gamma_2)^{-1}}\right).
\end{align}

\section{Proof of Theorem \ref{propstar5}.\ref{propstar5ii}} \label{App_propstar6}
We first show that $R_1^* \geq v(\{ 1\})$.
We have
\begin{subequations}
\begin{align}
&R_1^* (v) - v(\{1 \})  \nonumber \\
& =  \frac{1}{2} \log \left( \frac{1 +   h_{\Lambda}(\Gamma_1+ \Gamma_2)}{1  + h_1 \Gamma_1 + h_2 \Gamma_2} \right) -   \frac{1}{2} \left[ \frac{1}{2} \log \left( \frac{1 + 2  h_{\Lambda}\Gamma_2}{1 + 2 h_2 \Gamma_2} \right) \right] \nonumber \\
& \phantom{--}- \frac{1}{2} \log \left(\frac{1 +  \frac{\Gamma_1}{1+ (\sqrt{\Gamma_2}+\sqrt{\Lambda})^2} }{1+ h_1\Gamma_1} \right) \\
& \geq \frac{1}{2} \log \left( \frac{1 +   h_{\Lambda}(\Gamma_1+ \Gamma_2)}{1  + h_1 \Gamma_1 + h_2 \Gamma_2} \right) -   \frac{1}{2} \left[ \frac{1}{2} \log \left( \frac{1 + 2  h_{\Lambda}\Gamma_2}{1 + 2 h_2 \Gamma_2} \right) \right] \nonumber \\
& \phantom{--}- \frac{1}{2} \log \left(\frac{1 +  \frac{\Gamma_1}{1+ \Gamma_2+\Lambda} }{1+ h_1\Gamma_1} \right) \\
& = \frac{1}{4} \log \left[ \frac{(1 +   h_{\Lambda}(\Gamma_1+ \Gamma_2))^2}{(1  + h_1 \Gamma_1 + h_2 \Gamma_2)^2}  \frac{1 + 2 h_2 \Gamma_2}{1 + 2  h_{\Lambda}\Gamma_2}\right] \nonumber \\
&\phantom{--}+\frac{1}{4} \log \left[\frac{(1 +\Gamma_2+\Lambda)^2(1+ h_1\Gamma_1)^2}{(1 + \Gamma_1+ \Gamma_2+\Lambda )^2 }\right]  .
\end{align}
\end{subequations}
Next, we have
\begin{subequations}
\begin{align}
&\frac{(1 +   h_{\Lambda}(\Gamma_1+ \Gamma_2))^2}{(1  + h_1 \Gamma_1 + h_2 \Gamma_2)^2} \! \frac{1 + 2 h_2 \Gamma_2}{1 + 2  h_{\Lambda}\Gamma_2}\! \frac{(1 +\Gamma_2+\Lambda)^2(1+ h_1\Gamma_1)^2}{(1 + \Gamma_1+ \Gamma_2+\Lambda )^2 } \nonumber\\
&\geq 1 \\
&\Leftrightarrow \nonumber\\
&(1 +   h_{\Lambda}(\Gamma_1+ \Gamma_2))^2 (1 + 2 h_2 \Gamma_2) (1 +\Gamma_2+\Lambda)^2(1+ h_1\Gamma_1)^2 \nonumber \\
& -  (1  + h_1 \Gamma_1 + h_2 \Gamma_2)^2 (1 + 2  h_{\Lambda}\Gamma_2)(1 + \Gamma_1+ \Gamma_2+\Lambda )^2  \nonumber\\
&\geq 0\\
&\Leftrightarrow  \tfrac{1}{h_{\Lambda }^2}\Gamma _2 \left(\left(\Gamma _1+\Gamma
   _2\right) h_{\Lambda }+1\right){}^2 \nonumber \\
   & \phantom{l--}\times
   \left(\Gamma _1^2 h_1^2 \left(\Gamma _2
   h_{\Lambda }^2+2 h_2 \left(\Gamma _2
   h_{\Lambda }+1\right){}^2\right)\right. \nonumber \\
   & \left. \phantom{l--}+2 \Gamma
   _1 h_1 \left(\Gamma _2 h_{\Lambda }^2+h_2
   \left(2 \Gamma _2^2 h_{\Lambda }^2+2
   \Gamma _2 h_{\Lambda
   }+1\right)\right) \right. \nonumber\\
   & \left. \phantom{l--}+\Gamma _2
   \left(h_{\Lambda } - h_2\right) \left(h_2
   \left(2 \Gamma _2 h_{\Lambda
   }+1\right)+h_{\Lambda
   }\right)\right) \geq 0.
\end{align}
\end{subequations}
Since $h_{\Lambda } > h_2$, we deduce that $R_1^* \geq v(\{ 1\})$.

We now show that $R_2^* \geq v(\{ 2\})$.
We have
\begin{subequations}
\begin{align}
&R_2^* (v) - v(\{ 2\}) \nonumber\\ 
 & = \frac{1}{2} \left[ \frac{1}{2} \log \left( \frac{1 + 2  h_{\Lambda}\Gamma_2}{1 + 2 h_2 \Gamma_2} \right) \right] \nonumber \\
& \phantom{--}- \frac{1}{2} \log \left(\frac{1 +  \frac{\Gamma_2}{1+ (\sqrt{\Gamma_1}+\sqrt{\Lambda})^2} }{1+ h_2\Gamma_2} \right) \\
  & \geq \frac{1}{2} \left[ \frac{1}{2} \log \left( \frac{1 + 2  h_{\Lambda}\Gamma_2}{1 + 2 h_2 \Gamma_2} \right) \right] - \frac{1}{2} \log \left(\frac{1 +  \frac{\Gamma_2}{1+ \Gamma_1 + \Lambda} }{1+ h_2\Gamma_2} \right) \\
 & \geq \frac{1}{4}\log \left[ \frac{1 + 2  h_{\Lambda}\Gamma_2}{1 + 2 h_2 \Gamma_2}  \frac{(1 +\Gamma_1+\Lambda)^2(1+ h_2\Gamma_2)^2}{(1 + \Gamma_1+ \Gamma_2+\Lambda )^2 } \right].
\end{align}
\end{subequations}

Next, we have
\begin{subequations}
\begin{align}
&\frac{1 + 2  h_{\Lambda}\Gamma_2}{1 + 2 h_2 \Gamma_2}  \frac{(1 +\Gamma_1+\Lambda)^2(1+ h_2\Gamma_2)^2}{(1 + \Gamma_1+ \Gamma_2+\Lambda )^2 }  \geq 1\\
&\Leftrightarrow (1 + 2  h_{\Lambda}\Gamma_2)(1 +\Gamma_1+\Lambda)^2(1+ h_2\Gamma_2)^2 \nonumber\\
&\phantom{--}-(1 + 2 h_2 \Gamma_2)(1 + \Gamma_1+ \Gamma_2+\Lambda )^2  \geq 0\\
&\Leftrightarrow \tfrac{1}{h_{\Lambda }^2} \Gamma _2 \left(2 \Gamma _1^2
   h_{\Lambda }^3+2 \Gamma _1^2 \Gamma
   _2^2 h_2^2 h_{\Lambda }^3+4 \Gamma
   _1^2 \Gamma _2 h_2 h_{\Lambda }^3 \right. \nonumber \\
   & \left. \phantom{-}+4
   \Gamma _1 \Gamma _2^2 h_2^2 h_{\Lambda
   }^2+2 \Gamma _1 h_{\Lambda }^2+\Gamma
   _1^2 \Gamma _2 h_2^2 h_{\Lambda }^2 +4
   \Gamma _1 \Gamma _2 h_2 h_{\Lambda
   }^2\right.  \nonumber \\
  & \left. \phantom{-}+2 \Gamma _1 \Gamma _2 h_2^2
   h_{\Lambda }+\Gamma _2
   \left(h_2-h_{\Lambda }\right)
   \left(h_2 \left(2 \Gamma _2 h_{\Lambda
   }+1\right)+h_{\Lambda
   }\right)\right)\nonumber \\
  &\geq 0.
\end{align}
\end{subequations}
Similar to \eqref{eq87ff}, we have
\begin{align}
  \Gamma_2(h_2 - h_{\Lambda} ) \geq -  \Gamma_1 ( h_{\Lambda} - h_1 + 2 h_{\Lambda} \Gamma_2 (h_2 - h_1)).
\end{align}

We deduce
\begin{subequations}
\begin{align}
 & \tfrac{1}{h_{\Lambda }^2} \Gamma _2 \left(2 \Gamma _1^2
   h_{\Lambda }^3+2 \Gamma _1^2 \Gamma
   _2^2 h_2^2 h_{\Lambda }^3+4 \Gamma
   _1^2 \Gamma _2 h_2 h_{\Lambda }^3 \right. \nonumber \\
   & \left. \phantom{-}+4
   \Gamma _1 \Gamma _2^2 h_2^2 h_{\Lambda
   }^2+2 \Gamma _1 h_{\Lambda }^2+\Gamma
   _1^2 \Gamma _2 h_2^2 h_{\Lambda }^2+4
   \Gamma _1 \Gamma _2 h_2 h_{\Lambda
   }^2 \right. \nonumber \\
   & \left. \phantom{-}+2 \Gamma _1 \Gamma _2 h_2^2
   h_{\Lambda }+\Gamma _2
   \left(h_2-h_{\Lambda }\right)
   \left(h_2 \left(2 \Gamma _2 h_{\Lambda
   }+1\right)+h_{\Lambda
   }\right)\right)\nonumber\\
   & \geq \tfrac{1}{h_{\Lambda }^2} \Gamma _1 \Gamma _2 \left(2 \Gamma _1 \Gamma _2^2 h_2^2 h_{\Lambda }^3+2 \Gamma _1 h_{\Lambda }^3+4 \Gamma _1 \Gamma _2 h_2 h_{\Lambda }^3\right. \nonumber \\
   & \left. \phantom{-}+4 \Gamma _2^2 h_1 h_2 h_{\Lambda }^2+2 \Gamma _2 h_1 h_{\Lambda }^2+\Gamma _1 \Gamma _2 h_2^2 h_{\Lambda }^2\right. \nonumber \\
   & \left. \phantom{-}+4 \Gamma _2 h_1 h_2 
   h_{\Lambda }+h_{\Lambda }(h_{\Lambda }-h_2)+h_1 h_{\Lambda }+h_1 h_2\right) \\
   & \geq 0,
\end{align}
\end{subequations}
since $h_{\Lambda }>h_2$. We conclude that $R_2^* \geq v(\{ 2\})$.

\section{Proof of Lemma \ref{lemfunc}} \label{App_deriv}
For $ x\in\mathbb{R}_+$, we have \eqref{eqderiv1}, \eqref{eqderiv2}, and \eqref{eqderiv3}. For $ x\in [0,c[$, we have \eqref{eqderiv4}.
\begin{figure*}[b!]
\hrulefill
\begin{align} \label{eqderiv1}
 f^{(1)'}_{k,h_1,h_2,a} (x) &=  \frac{ - \left(h_2-h_1\right) k}{\left(h_1 (a+k x)+1\right) \left(h_2 (a+k x)+1\right)} <0,\\
f^{(2)'}_{k,h_1,h_2,a,b} (x) &= \frac{-\left(h_2-h_1\right) k (b-a) \left(h_1 \left(h_2 (a k+b k+2
   x)+1\right)+h_2\right)}{\left(h_1 (a k+x)+1\right) \left(h_2 (a k+x)+1\right) \left(h_1 (b
   k+x)+1\right) \left(h_2 (b k+x)+1\right)} \leq 0,\label{eqderiv2}\\
f^{(3)'}_{k,h_1,h_2,a}(x)  
&=  \frac{ - \left(h_2-h_1\right) k (k+1) x \left(h_1 \left(h_2 (2 a+2 k
   x+x)+1\right)+h_2\right)}{\left(h_1 (a+k x)+1\right) \left(h_1 (a+k x+x)+1\right) \left(h_2
   (a+k x)+1\right) \left(h_2 (a+k x+x)+1\right)}\leq 0,\label{eqderiv3}\\
   f^{(4)'}_{k,h_1,h_2,a,c} (x)  
& = \frac{- \left(h_2-h_1\right) k (k+1) (c-x) \left(h_1 \left(h_2 (2 a+c k+(k+2)
   x)+1\right)+h_2\right)}{\left(h_1 (a+k x+x)+1\right) \left(h_2 (a+k x+x)+1\right) \left(h_1
   (a+c k+x)+1\right) \left(h_2 (a+c k+x)+1\right)} \leq 0,\label{eqderiv4}\\
   f_{h_1,h_2,a,b}^{(6)'} (\omega) &= \frac{-\left(h_2-h_1\right) \left(\frac{(a+b) \log \left(\frac{a h_2+\omega }{a h_1+\omega }\right)}{\left(h_1
   (a+b)+\omega \right) \left(h_2 (a+b)+\omega \right)}-\frac{a \log \left(\frac{h_2 (a+b)+\omega }{h_1
   (a+b)+\omega }\right)}{\left(a h_1+\omega \right) \left(a h_2+\omega \right)}\right)}{\log ^2\left(\frac{a
   h_2+\omega }{a h_1+\omega }\right)},\label{eqderiv6a}\\
   f_{h_1,h_2,a,b}^{(6)'} (\omega) &= \frac{-\left(h_2-h_1\right)(a+b) a \left(g(a) - g(a+b)\right) }{\left(a h_1+\omega \right) \left(a h_2+\omega \right) \left(h_1
   (a+b)+\omega \right) \left(h_2 (a+b)+\omega \right)\log ^2\left(\frac{a
   h_2+\omega }{a h_1+\omega }\right)}.\label{eqderiv6b}
   \end{align}
\begin{align}
A(\omega) &
\triangleq  \left(\frac{b \left(a h_1
   h_2 (a+b)-\omega ^2\right) \log
   \left(\frac{\left(h_1 (a+b+c)+\omega \right)
   \left(h_2 (a+b+c+d)+\omega \right)}{\left(h_2
   (a+b+c)+\omega \right) \left(h_1 (a+b+c+d)+\omega
   \right)}\right)}{\left(a h_1+\omega \right)
   \left(a h_2+\omega \right) \left(h_1 (a+b)+\omega
   \right) \left(h_2 (a+b)+\omega \right)}  \right.  \nonumber \\
   &  \left. \phantom{-----} -\frac{d   \left(h_1 h_2 (a+b+c) (a+b+c+d)-\omega
   ^2\right)\log \left(\frac{\left(a h_1+\omega \right)
   \left(h_2 (a+b)+\omega \right)}{\left(a h_2+\omega
   \right) \left(h_1 (a+b)+\omega \right)}\right) }{\left(h_1 (a+b+c)+\omega \right)
   \left(h_2 (a+b+c)+\omega \right) \left(h_1
   (a+b+c+d)+\omega \right) \left(h_2
   (a+b+c+d)+\omega \right)}\right), \label{eqAomega} \\
   A(\omega) &\leq \frac{C(\omega) \left[h_1 \left(D(\omega)+a c h_2^2 (a+b+c+d)+c \omega ^2\right)+h_1^2 h_2 (a+b) (a+b+c) (b+c+d)+h_2 \omega ^2 (b+c+d)\right]}{B(\omega)} \leq 0 \label{eqcase1}, \\
      A(\omega) &\leq \frac{ C(\omega) \left[h_1 \left(D(\omega)+h_2^2 (a+b) (a+b+c) (b+c+d)+\omega^2 (b+c+d)\right)+a c h_1^2 h_2 (a+b+c+d)+c h_2 \omega ^2\right]}{B(\omega)} \leq 0 \label{eqcase2}, \\
      A(\omega) &\leq \frac{C(\omega) \! \left[h_1 \left(D(\omega)+h_2^2 (a+b) (b+c) (a+b+c+d)+\omega^2 (b+c)\right)+a h_1^2 h_2 (c+d) (a+b+c)+h_2 \omega^2 \! (c+d)\right]}{B(\omega)} \leq 0 \label{eqcase3}, \\
      A(\omega) &\leq \frac{ C(\omega) \! \left[h_1 \left(D(\omega)+a h_2^2 (c+d) (a+b+c)+\omega^2 \! (c+d)\right)+h_1^2 h_2 (a+b) (b+c) (a+b+c+d)+h_2 \omega^2 \! (b+c)\right]}{B(\omega)} \leq 0 \label{eqcase4} .
\end{align}
\end{figure*}

\section{Proof of Lemma \ref{lemfunc2}} \label{App_deriv2}

To prove that the two functions in Lemma \ref{lemfunc2} are non-decreasing, we show that their derivatives are non-negative. 

For $\omega \in \mathbb{R}_+$, we have
\begin{align}
&f_{h_1,h_2,a,b,c,d}^{(5)'}(\omega) = \frac{-\left(h_2-h_1\right)}{\log
   ^2\left(\frac{\left(a h_1+\omega \right) \left(h_2
   (a+b)+\omega \right)}{\left(a h_2+\omega \right)
   \left(h_1 (a+b)+\omega \right)}\right)} \times A(\omega),
    \end{align}
  where $A(\omega)$ is defined in \eqref{eqAomega}.

It is thus sufficient to show that $A(\omega)\leq 0$ to obtain $f_{h_1,h_2,a,b,c,d}^{(5)'}(\omega) \geq 0$. We  consider four cases and define for convenience 
\begin{align*}
B(\omega)&\triangleq \left(a
   h_1+\omega \right) \left(a h_2+\omega \right) \left(h_1 (a+b)+\omega \right)  \\
&\phantom{--}\times\left(h_2 (a+b)+\omega \right)\left(h_1 (a+b+c)+\omega \right) \\
&\phantom{--}\times   \left(h_2 (a+b+c)+\omega \right)  \left(h_1 (a+b+c+d)+\omega \right)
   \\
&\phantom{--}\times\left(h_2 (a+b+c+d)+\omega \right),\\
   C(\omega)&\triangleq -b d \left(h_2-h_1\right) \omega ^2,\\
   D(\omega)&\triangleq 2 h_2 \omega  (a (b+2 c+d)+(b+c) (b+c+d)).
   \end{align*}

\textbf{Case 1}: Assume that $d   (h_1 h_2 (a+b+c) (a+b+c+d)-\omega
   ^2) \geq 0$ and $b \left(a h_1
   h_2 (a+b)-\omega ^2\right) \geq 0$.
Using that $\forall x \in \mathbb{R}_+^*, 1-x^{-1}\leq \log(x) \leq x-1$, we have \eqref{eqcase1}.

\textbf{Case 2}: Assume that $d   (h_1 h_2 (a+b+c) (a+b+c+d)-\omega
   ^2) < 0$ and $b \left(a h_1
   h_2 (a+b)-\omega ^2\right) < 0$.
Using that $\forall x \in \mathbb{R}_+^*, 1-x^{-1}\leq \log(x) \leq x-1$, we have \eqref{eqcase2}.

\textbf{Case 3}: Assume that $d   (h_1 h_2 (a+b+c) (a+b+c+d)-\omega
   ^2) < 0$ and $b \left(a h_1
   h_2 (a+b)-\omega ^2\right) \geq 0$.
Using that $\forall x \in \mathbb{R}_+^*,  \log(x) \leq x-1$, we have \eqref{eqcase3}.

\textbf{Case 4}: Assume that $d   (h_1 h_2 (a+b+c) (a+b+c+d)-\omega
   ^2) \geq 0$ and $b \left(a h_1
   h_2 (a+b)-\omega ^2\right) < 0$.
Using that $\forall x \in \mathbb{R}_+^*, 1-x^{-1}\leq \log(x) $, we have \eqref{eqcase4}.

Next, for $\omega \in \mathbb{R}_+$, we have \eqref{eqderiv6a}, which we rewrite as \eqref{eqderiv6b}, 
where we have defined
\begin{equation}
g : \mathbb{R_+^*} \to \mathbb{R}, x \mapsto \frac{(x h_1 + \omega)(x h_2 + \omega) \log \left( \frac{x h_2 + \omega}{xh_1 + \omega} \right)}{x}.
\end{equation}
   
Hence, to show that $f_{h_1,h_2,a,b}^{(6)'} (\omega) \geq 0 $, it is sufficient to show that $g$ is non-decreasing. We have for $x \in \mathbb{R_+^*}$,
\begin{equation}
g'(x) =\frac{\left(h_1 h_2 x^2-\omega ^2\right) \log \left(\frac{h_2 x+\omega }{h_1 x+\omega }\right)+ x\omega (h_2 -h_1) }{x^2}.
\end{equation}

We now show that $g'(x) \geq 0$. We consider two cases.

\textbf{Case 1}: Assume $\left(h_1 h_2 x^2-\omega ^2\right) <0$. Using that $\forall x \in \mathbb{R}_+^*, \log(x) \leq x-1$, we lower bound $g'(x)$ by
\begin{equation}
\frac{h_1 \left(h_2-h_1\right) \left(h_2 x+\omega \right)}{h_1 x+\omega } \geq 0.
\end{equation}

\textbf{Case 2}: Assume $\left(h_1 h_2 x^2-\omega ^2\right) \geq 0$. Using that $\forall x \in \mathbb{R}_+^*, 1-x^{-1}\leq \log(x)$, we lower bound $g'(x)$ by
\begin{equation}
\frac{\left(h_2-h_1\right) h_2 \left(h_1 x+\omega \right)}{h_2 x+\omega } \geq 0.
\end{equation}

\bibliographystyle{IEEEtran}
\bibliography{bib}

\end{document}